%% file: main.tex
\documentclass[letterpaper, 10 pt, conference]{ieeeconf}
\IEEEoverridecommandlockouts
\newcommand{\comment}[1]{}

\newcommand{\be}[0]{\begin{equation}}
\newcommand{\ee}[0]{\end{equation}}
\newcommand{\ben}[0]{\begin{equation*}}
\newcommand{\een}[0]{\end{equation*}}
\newcommand{\bena}[0]{\begin{eqnarray*}}
\newcommand{\eena}[0]{\end{eqnarray*}}
\newcommand{\bea}[0]{\begin{eqnarray}}
\newcommand{\eea}[0]{\end{eqnarray}}

\usepackage[english]{babel}
\usepackage{amsthm}
\usepackage[version=4]{mhchem}
\usepackage{siunitx}
\usepackage{array}
\usepackage{longtable}
\usepackage{nomencl}
\usepackage{mathrsfs} 
\usepackage{mathtools} 
\usepackage[list=true]{subcaption}
\usepackage{optidef}
\usepackage{multicol}
\usepackage{svg}
\usepackage{lipsum}
\usepackage{calrsfs}
\usepackage{scalerel}
\usepackage{stackengine}
\usepackage{booktabs}
\usepackage{caption}
\usepackage{enumerate}
\usepackage{cite}
\usepackage{amssymb,amsfonts}
\usepackage{algorithmic}
\usepackage{graphicx}
\usepackage{textcomp}

\DeclareMathAlphabet{\mathcal}{OMS}{cmsy}{m}{n}

\theoremstyle{definition}
\newtheorem{definition}{Definition}
\newtheorem{theorem}{Theorem}

\newtheorem{remark}{Remark}

\begin{document}

\title{Closed-Loop Model-Based Control Barrier Functions with Application to Robust Flight Envelope Protection}
\author{Johannes Autenrieb$^{1}$, Mark Spiller$^{1}$, Peter A. Fisher$^{2}$, Junhyeok Yoon$^{2}$,\\ Spilios Theodoulis$^{3}$ and  Anuradha Annaswamy$^{2}$
\thanks{$^{1}$ German Aerospace Center (DLR), Institute of Flight Systems, 38108 Braunschweig, Germany.
(email: \texttt{johannes.autenrieb@dlr.de, mark.spiller@dlr.de})}
\thanks{$^{2}$ Massachusetts Institute of Technology, Department of Mechanical Engineering, Cambridge, MA 02139, USA.
(email: \texttt{pafisher@mit.edu,yjunh01@mit.edu, aanna@mit.edu})}
\thanks{$^{3}$ Delft University of Technology, Faculty of Aerospace Engineering, 2628 Delft, The Netherlands.
(email: \texttt{s.theodoulis@tudelft.nl})}
}

\maketitle

\begin{abstract}
\input{Sections/abstract}
\end{abstract}

\section{Introduction}
\label{sec:Introduction}
\input{Sections/Introduction}

\section{Preliminaries}
\label{sec:Preliminaries}
\input{Sections/Preliminaries}

\section{Problem Formulation}
\label{sec:Problem Formulatio}
\input{Sections/Problem_Formulation}

\section{Input-Level Safety Filters \& Altered Closed-Loop Dynamics}
\label{sec:Motivating Example}
\input{Sections/Motivation}

\section{Safe Reference Signal Generation via Closed-Loop Model-Based Control Barrier Functions}
\label{sec:Closed_Loop_Model_CBFs}
\input{Sections/Closed_Loop_Model_CBFs}

\section{Application of CLM-CBFs on Flight Envelope Protection for Missile Systems}
\label{sec:Application of CLM-CBFs on Flight Envelope Protection for Missile Systems}
\input{Sections/Application}

\section{Conclusion}
\label{sec:Conclusion}
\input{Sections/Conclusion}

\bibliographystyle{IEEEtran}
\bibliography{./sample}

\end{document}

%% file: Sections/abstract.tex
Ensuring operation of aerospace systems within prescribed flight envelope limits is a fundamental requirement for modern flight control architectures. Flight envelope protection aims to prevent violations of aerodynamic and structural constraints, thereby mitigating risks such as stall and excessive load factors. Control barrier functions (CBFs) have emerged as a principled tool for enforcing safety by ensuring that the system state remains within a prescribed safe set. In most existing approaches, safety constraints are imposed at the control input-level based on an open-loop model of the system. While this open-loop model-based CBF formulation enables modular design, it may alter the closed-loop system dynamics, potentially compromising robustness guarantees and complicating integration into existing flight control architectures. This paper proposes a closed-loop model-based control barrier function (CLM-CBF) framework for flight envelope protection. The key idea is to enforce safety at the reference-level using an explicit model of the closed-loop system, thereby preserving the stability and robustness properties of the underlying controller. This formulation enables safety filtering without modifying the control law, facilitating modular integration and retrofitting into existing systems. 

%% file: Sections/Introduction.tex
Recent advances in aerospace engineering have led to novel vehicle configurations and expanding application domains across both civil and defense sectors. Modern aerospace systems, such as high-speed missiles, hypersonic glide vehicles, and next-generation autonomous aircraft, operate at the limits of aerodynamics, propulsion, and structural integrity, under extreme conditions including high dynamic pressures, aggressive maneuvers, and limited control authority. 

In this context, Flight Envelope Protection (FEP) plays a central role by ensuring that the system remains within admissible operational limits defined by aerodynamic, structural, and performance constraints. The flight envelope is typically characterized by quantities such as velocity, load factor, altitude, and angle of attack, which must be respected to prevent critical failure modes including structural overload, aerodynamic stall, or loss of control \cite{Oudin2017,Stougie2024-ym}. More generally, FEP represents a class of safety-critical control problems that can be interpreted as maintaining the system state within a prescribed admissible set, a perspective that also arises in related applications such as collision avoidance and cooperative guidance.

\begin{figure}
\centering
\includegraphics[width=\columnwidth]{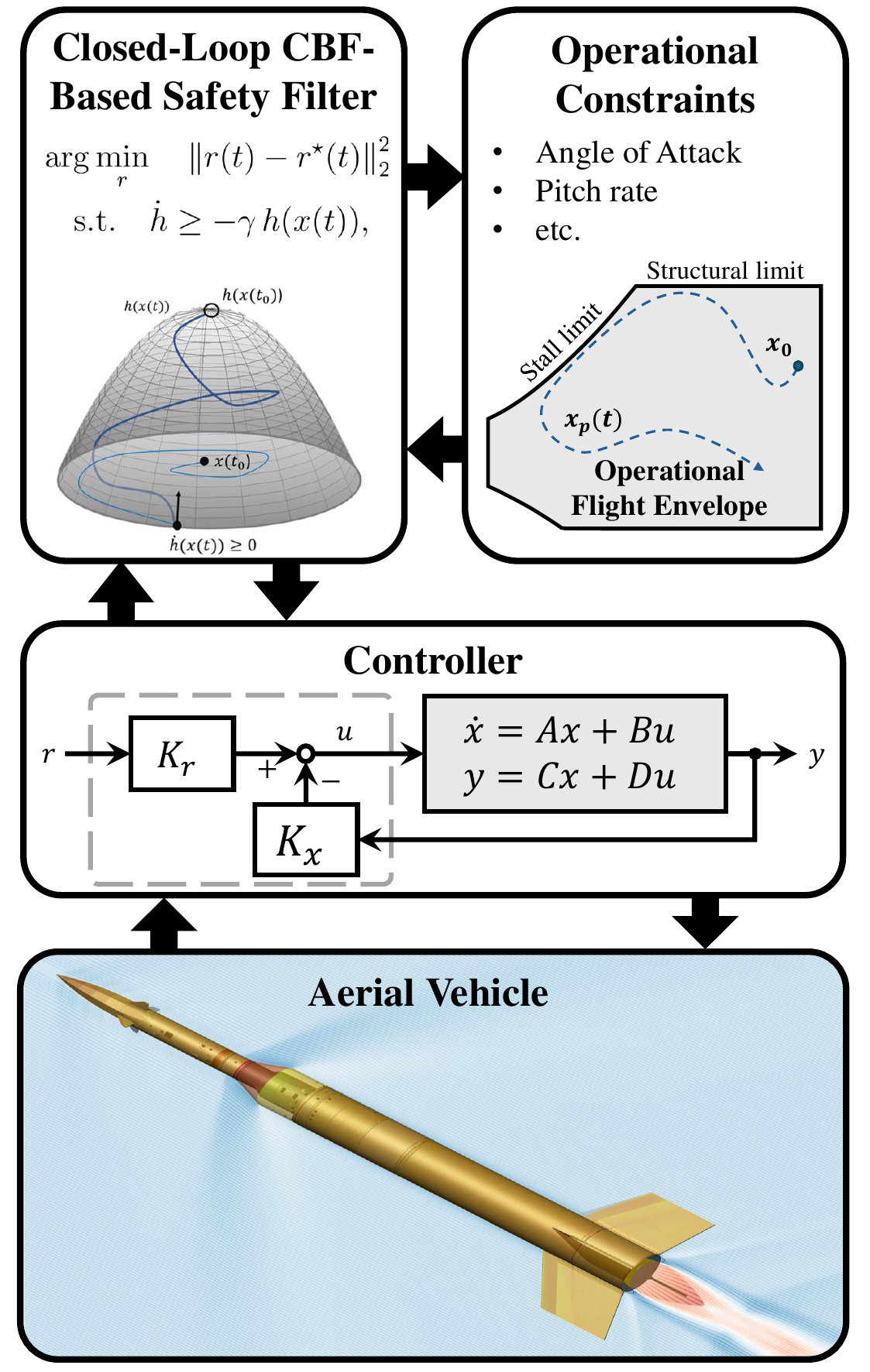}
\caption{Illustration of the proposed closed-loop model-based control barrier function safety filter concept.}
\label{fig:overview_concept}
\end{figure}

In current aerospace practice, FEP is commonly implemented at the attitude control level using reference filtering techniques such as reference and command filtering \cite{Tang2009,Falkena2010,Seo2017,Lombaerts2017}. Reference clipping enforces safety by saturating commanded signals within predefined bounds, but neglects the influence of closed-loop system dynamics, which may lead to constraint violations during transient phases and offers no formal guarantees under disturbances or model mismatch \cite{Stougie2024-ym,Steffensen2019}. Model-based command filtering partially addresses these limitations by incorporating system dynamics into the reference shaping and can provide improved transient behavior and, under certain assumptions, forward invariance guarantees \cite{Grondman2018,Autenrieb_2024b,Autenrieb2024}. However, these approaches are sensitive to modeling inaccuracies and external disturbances. As a consequence, conservative tuning is often required to ensure safety, which directly limits achievable performance, particularly for high-speed aerospace systems \cite{Falkena2010}.

Control Barrier Functions (CBFs) have emerged as a systematic framework for enforcing state constraints through forward invariance \cite{Ames_2014, Autenrieb2023b, Fisher2026,  Autenrieb2025ACC}. In most practical implementations, CBFs are realized via optimization-based safety filters formulated as quadratic programs, where a performance-oriented controller generates a nominal input that is subsequently modified to satisfy the safety constraint. Since these constraints are typically derived using the open-loop system dynamics, this paradigm can be interpreted as an open-loop model-based CBF (OLM-CBF) approach. A key advantage of OLM-CBF methods is their modularity, as the safety filter can be designed independently of the controller structure. Despite their theoretical appeal, OLM-CBF approaches exhibit significant limitations in aerospace flight control applications. In particular, modifying the control input at runtime may alter the closed-loop behavior of the system and compromise carefully designed properties such as stability margins, transient performance, and robustness guarantees. For high-performance flight control systems, where controllers are tuned to meet stringent specifications, such modifications can lead to undesirable characteristics. Furthermore, input-level safety filters may interact adversely with internal controller dynamics, for instance by inducing integrator windup or disrupting gain-scheduled architectures. From a practical perspective, these effects complicate the integration of OLM-CBF-based safety filters into existing, certified flight control systems, where modifications of the closed-loop dynamics are costly and difficult to validate.

Motivated by these limitations, this work adopts a fundamentally different perspective on safety-critical control. Instead of enforcing safety at the control input-level using open-loop dynamics, we propose to enforce safety at the reference-level using an explicit model of the closed-loop system. This leads to the concept of a {closed-loop model-based Control Barrier Function} (CLM-CBF) framework. In this formulation, the safety filter operates on the commanded reference signal and leverages the known closed-loop mapping from reference to state, allowing safety constraints to be evaluated with respect to the actual closed-loop behavior. This shift in perspective provides two key advantages. First, the stabilizing structure and performance characteristics of the underlying controller, such as gain and phase margins, transient response, and robustness properties, are preserved, since the generated signal of the nominal control law itself is not modified. Second, the approach is naturally compatible with existing flight control architectures, enabling modular integration and facilitating retrofitting without requiring redesign or re-certification of the baseline controller. 

The proposed approach is applied to the flight envelope protection problem of a longitudinal missile system with realistic aerodynamic characteristics and actuator constraints. The resulting safety filter is formulated as a quadratic program that computes the least-invasive modification of the commanded reference signal while ensuring strict state constraint satisfaction. Numerical results demonstrate that the proposed CLM-CBF approach reliably prevents constraint violations during aggressive maneuvers, while preserving robustness guarantees.

%% file: Sections/Preliminaries.tex
Consider a control-affine nonlinear system 
\begin{align}
    \dot{\mathbf{x}(t)} &= \mathbf{f}(\mathbf{x}(t)) + \mathbf{g}(\mathbf{x}(t)) \mathbf{u}(t), \label{NonlinearPlant1}
\end{align}
where \( \mathbf{x}(t) \in \chi \subset R^n\), \( \mathbf{u}(t) \in U \subset R^m \), \( \mathbf{y} \in \mathbb{R}^p \), and \( \mathbf{f}: \chi \to \mathbb{R}^n \) and \( \mathbf{g}: \chi \to \mathbb{R}^{n \times m} \) are sufficiently smooth functions. To define safety, we consider a continuously differentiable function \( h: \chi \rightarrow \mathbb{R} \) and a set \( S \) defined as the zero-superlevel set of \( h \), yielding:
\begin{equation}
    \label{Safe_set_1}
    S \triangleq \big\{ \mathbf{x}(t) \in \chi \mid h(\mathbf{x}(t)) \geq 0 \big\}.
\end{equation}



We introduce the notion of a {control barrier function} (CBF) such that its existence allows the system to be rendered safe with respect to \( S \) using a control input \( \mathbf{u}(t) \) \cite{ames2016control}.
\begin{definition}[CBF, \cite{ames2016control}]
Let \( S \subset \chi \) be the zero-superlevel set of a continuously differentiable function \( h: \chi \rightarrow \mathbb{R} \). The function \( h \) is a CBF for \( S \) for all \( \mathbf{x}(t) \in S \), if there exists a class \( \mathcal{K}_{\infty} \) function \( \alpha(h(\mathbf{x}(t))) \) such that for the dynamics defined in \eqref{NonlinearPlant1} we obtain:
\begin{equation}
    \label{ControlBarrierFunction_simple}
    \sup_{\mathbf{u}(t)\in U} [L_{\mathbf{f}} h(\mathbf{x}(t)) + L_{\mathbf{g}} h(\mathbf{x}(t)) \mathbf{u}(t)]  \geq -\alpha(h(\mathbf{x}(t))),
\end{equation}
\end{definition}

\begin{theorem}
\label{theorem_LCBF}
Given a set \( S \subset \chi \), defined via the associated CBF as in \eqref{Safe_set_1}, any Lipschitz continuous controller \( \mathbf{k}(\mathbf{x}(t)) \in K_{S}(\mathbf{x}(t)) \) with 
\begin{equation*}
\begin{aligned}
    K_{S} (\mathbf{x}(t)) = \big\{ \mathbf{u}(t) \in U :& L_{\mathbf{f}} h(\mathbf{x}(t)) +\\ &L_{\mathbf{g}} h(\mathbf{x}(t)) \mathbf{u}(t) + \alpha(h(\mathbf{x}(t))) \ge 0 \big\}
\end{aligned}
\label{definition_safe_controller}
\end{equation*}
renders the system \eqref{NonlinearPlant1} forward invariant within \( S \) \cite{XU2015}. 

One way to construct a controller satisfying \eqref{definition_safe_controller} is through a quadratic program-based safety filter applied at the control input-level, as proposed in~\cite{Ames_2014}: 
\begin{equation}
\begin{aligned}
\mathbf{u}(t) = \arg\min_{\mathbf{u}(t)\in U} \quad
& \|\mathbf{u}(t) - \mathbf{u}(t)^\star\|_2^2 \\
\text{s.t.}\quad
& L_{\mathbf{f}} h(\mathbf{x}(t))+\\
& L_{\mathbf{g}} h(\mathbf{x}(t))\,\mathbf{u}(t)
  \ge -\alpha\!\left(h(\mathbf{x}(t))\right),
\end{aligned}
\label{eq:qp_cbf}
\end{equation}
where $\mathbf{u}(t)^\star$ is a performance-oriented, but potentially not safe, control input.
\end{theorem}

The formulation in \eqref{eq:qp_cbf} represents the standard paradigm for CBF-based safety filtering, where safety is enforced at the control input-level. In this setting, a nominal or performance-oriented controller generates a desired input $\mathbf{u}(t)^\star$, which is subsequently modified by a quadratic program to ensure satisfaction of the safety constraint. A key characteristic of this approach is that the barrier condition is derived solely from the open-loop system dynamics, i.e., the vector fields $\mathbf{f}(\mathbf{x}(t))$ and $\mathbf{g}(\mathbf{x}(t))$, without explicitly accounting for the structure or dynamics of the underlying controller. Consequently, the safety filter operates independently of the closed-loop system behavior and treats the controller as a black-box input generator. In the remainder of this paper, we refer to this class of methods as open-loop model-based CBFs (OLM-CBFs). A conceptual illustration of this architecture is shown in Fig.~\ref{fig:OLM_CBF_architecture}.

\begin{figure*}
    \centering
    \includegraphics[width=0.8\textwidth]{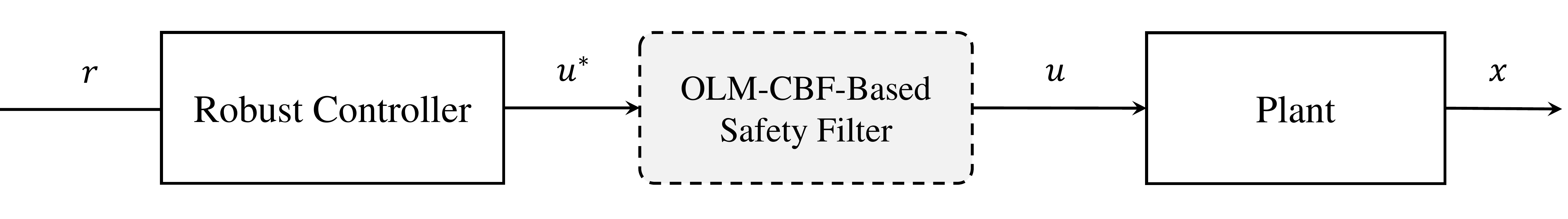}
    \caption{Conceptual control architecture with an open-loop model-based Control Barrier Function (OLM-CBF).}
    \label{fig:OLM_CBF_architecture}
\end{figure*}

%% file: Sections/Problem_Formulation.tex
\subsection{Longitudinal Missile Dynamics and Design Model}
\label{Sect: Nonlinear_Missile_Model}


This study considers an axially symmetric skid-to-turn tail-controlled missile system. Under the assumption of a rapidly stabilized roll channel, the missile motion can be decomposed into two perpendicular and weakly coupled channels (pitch and yaw). In the following, we focus on the longitudinal (pitch) dynamics, which capture the dominant behavior relevant for control design and flight envelope protection.  The missile pitch dynamics model utilized in this study is an adjusted version of the dynamical model presented in \cite{Hwang_2017}, obtained from a full six-degree-of-freedom formulation by removing non-relevant dynamics. The missile dynamics are governed by nonlinear aerodynamic forces and moments, which depend on the operating condition
\begin{equation}
\sigma = \begin{bmatrix} M & h \end{bmatrix}^\top,
\end{equation}
where $M$ denotes the Mach number and $h$ the altitude. For the considered flight regime, the dominant dynamics are given by the short-period motion, which can be described by the angle of attack $\alpha(t)$ and pitch rate $q(t)$. The corresponding nonlinear dynamics are given by
\begin{align}
    \dot{\alpha}(t) &= \frac{QS}{mV} \left( C_{Z,\alpha}(\sigma)\alpha(t)+C_{Z,\delta}(\sigma)\delta(t)\right) + q(t), \label{eqn:alpha_dot_clean} \\
    \dot{q}(t) &= \frac{QSd}{I_{yy}} \Big(
        C_{M,\alpha}(M)\alpha(t) \notag \\
        &\quad + C_{M,q}(M)\frac{d}{2V}q(t)
        + C_{M,\delta}(M)\delta(t)
    \Big), \label{eqn:q_dot_clean}
\end{align}
where $\delta(t)$ denotes the control fin deflection. The dynamic pressure is given by $Q = \frac{1}{2}\rho V^2$, and the parameters $S$, $d$, $m$, and $I_{yy}$ denote the reference area, diameter, mass, and moment of inertia, respectively. The normal load factor is given by
\begin{equation}
    n_z(t)= \frac{QS}{mg} \left(C_{Z,\alpha}(\sigma)\alpha(t)+C_{Z,\delta}(\sigma)\delta(t)\right),
    \label{eqn:nz_clean}
\end{equation}
which highlights the strong dependence of aerodynamic loads on the angle of attack. For control design, the nonlinear dynamics are linearized around a trim condition $(x_0, \delta_0)$ corresponding to a constant operating point $\sigma_0 = [M_0, h_0]^\top$. Defining perturbation variables
\begin{equation}
x(t) = \begin{bmatrix} \alpha(t) - \alpha_0 \\ q(t) - q_0 \end{bmatrix}, 
\quad
u(t) = \delta(t) - \delta_0,
\end{equation}
the local dynamics are approximated by the linear time-invariant system
\begin{equation}
\label{eq:LTI_short_period}
\dot{x}(t) = A(\sigma_0)x(t) + B(\sigma_0)u(t),
\end{equation}
with
\begin{equation}
A(\sigma_0) =
\begin{bmatrix}
a_{11} & 1 \\
a_{21} & a_{22}
\end{bmatrix},
\quad
B(\sigma_0) =
\begin{bmatrix}
b_1 \\
b_2
\end{bmatrix},
\end{equation}
where the coefficients are given by
\begin{align}
a_{11} &= \frac{QS}{mV} C_{Z,\alpha}(\sigma_0), \quad a_{21} = \frac{QSd}{I_{yy}} C_{M,\alpha}(\sigma_0), \notag \\
a_{22} &= \frac{QSd}{I_{yy}} C_{M,q}(\sigma_0)\frac{d}{2V}, \quad b_1 = \frac{QS}{mV} C_{Z,\delta}(\sigma_0), \notag \\
b_2 &= \frac{QSd}{I_{yy}} C_{M,\delta}(\sigma_0).\notag 
\end{align}

Further, we define our control variable as $\alpha(t)$, resulting in the corresponding control error
\begin{equation}
e(t)= 
r(t)-\alpha(t).
\label{eq:tracking_error}
\end{equation}

\medskip

The linear model \eqref{eq:LTI_short_period} represents the short-period dynamics of the missile and forms the basis for the subsequent controller synthesis and safety filter design. It captures the dominant coupling between angle of attack and pitch rate, which governs both the dynamic response and the aerodynamic load behavior. The numerical values of the aerodynamic coefficients and vehicle parameters used in this study are summarized in Table~\ref{tab:missile_params}. In the present study, the altitude is assumed to be constant at mean sea level (MSL), such that variations in the operating condition are primarily governed by the Mach number. The control input is subject to actuator magnitude and rate constraints, which must be respected during control design and implementation.

\begin{table}[hbt!]
\caption{\label{tab:missile_params} Flight vehicle parameters, aerodynamic coefficients, and constraints.}
\centering
\begin{tabular}{lcc}
\hline
Parameter & Symbol & Value \\
\hline
Velocity & $V$ & 914.00 m/s \\
Air density & $\rho$ & 1.225 kg/m$^3$ \\
Mass & $m$ & 453.00 kg \\
Moment of inertia & $I_{yy}$ & 1407.00 kgm$^2$ \\
Gravity & $g$ & 9.80665 m/s$^2$ \\
Mach number & $M$ & 2.6859 \\
Reference area & $S$ & 0.073 m$^2$ \\
Missile diameter  & $d$ & 0.30 m \\
\hline
Aerodynamic Coefficients & & \\
\hline
$C_{Z,\alpha}$ & & -32.5925 \\
$C_{Z,\delta}$ & & -7.1863 \\
$C_{m,\alpha}$ & & -80.4716 \\
$C_{m,q}$ & & -56.1499 \\
$C_{m,\delta}$ & & -69.6272 \\
\hline
Control Input Constraints & & \\
\hline
Fin deflection limits & $\delta$ & $\pm 30^\circ$ \\
Fin deflection rate limits & $\dot{\delta}$ & $\pm 90^\circ$/s \\
\hline
\end{tabular}
\end{table}


\subsection{Flight Envelope Protection for Longitudinal Missile Dynamics}
\label{sec:Problem_Formulation_Missile}

Based on the longitudinal short-period dynamics introduced in Section~\ref{Sect: Nonlinear_Missile_Model}, the control objective is to track a desired reference signal while ensuring safe operation within an a priori defined flight envelope. As indicated by the nonlinear load expression in \eqref{eqn:nz_clean}, the normal load factor $n_z(t)$ is directly influenced by the angle of attack $\alpha(t)$ and depends on the airspeed $V$ through the dynamic pressure. Consequently, the admissible operating region of the missile is characterized by flight-speed-dependent load constraints arising from aerodynamic and structural limitations. For the considered longitudinal dynamics, the angle of attack $\alpha(t)$ constitutes the dominant safety-critical variable. In particular, since $n_z(t)$ can be expressed as an approximately affine function of $\alpha(t)$, the velocity-dependent load constraints can be conservatively reformulated as bounds on $\alpha(t)$ \cite{Autenrieb_2025}. This allows the flight envelope constraints to be expressed directly in terms of the system state. Accordingly, the FEP problem is formulated as maintaining
\begin{equation}
\label{eq:alpha_constraint}
\alpha_{\min}(\sigma_0) \leq \alpha(t) \leq \alpha_{\max}(\sigma_0),
\end{equation}
where the bounds depend on the operating point $\sigma_0$ and capture aerodynamic and structural limitations. Defining the state vector $\mathbf{x}(t) = [\alpha(t), q(t)]^\top$, the admissible set can be written as
\begin{equation}
\mathcal{S}(\sigma_0)
\triangleq
\left\{
\mathbf{x}(t) \in \mathbb{R}^2 \;\mid\;
\alpha_{\min}(\sigma_0) \leq \alpha(t) \leq \alpha_{\max}(\sigma_0)
\right\}.
\end{equation}
The objective of flight envelope protection is to ensure forward invariance of $\mathcal{S}(\sigma_0)$, i.e.,
\begin{equation}
\mathbf{x}(t_0) \in \mathcal{S}(\sigma_0)
\;\Rightarrow\;
\mathbf{x}(t) \in \mathcal{S}(\sigma_0), \quad \forall t \geq t_0,
\end{equation}
under the closed-loop dynamics resulting from the control design applied to the linearized system \eqref{eq:LTI_short_period}.\\


The problem addressed in this work can therefore be stated as follows: Given the linearized short-period model \eqref{eq:LTI_short_period} with magnitude and rate constraints on the actuators, design a control architecture that ensures tracking performance while guaranteeing forward invariance of the safe set $\mathcal{S}(\sigma_0)$. As shown in the motivating example in Section~\ref{sec:Motivating Example}, input-level CBF-based safety filtering modifies the effective closed-loop dynamics whenever the constraint is active, such that the system is no longer governed by $A_{\mathrm{cl}}(\sigma_0)$. Consequently, nominal stability and performance guarantees are not preserved in general. Motivated by this limitation, the objective is to enforce \eqref{eq:alpha_constraint} while preserving the nominal closed-loop behavior. To this end, a performance-oriented baseline controller is combined with a safety filter acting at the reference-level, modifying $\mathbf{r}(t)$ rather than $\mathbf{u}(t)$.

%% file: Sections/Motivation.tex
While the previous section introduced the standard formulation of control barrier functions and their realization as input-level safety filters, this paradigm inherently relies on an open-loop description of the system dynamics. In particular, the resulting safety constraints are constructed independently of the nominal closed-loop behavior induced by the controller. In the following, we highlight fundamental limitations of the OLM-CBF approach, as defined in \eqref{eq:qp_cbf}, when combined with feedback controllers. To this end, a simple illustrative example is used to demonstrate how the interaction between the safety filter and the nominal controller can induce nontrivial, state-dependent modifications of the closed-loop dynamics, thereby potentially degrading robustness and performance guarantees of the nominal controller.

We consider the following linear time-invariant (LTI) system representing a plant with partially unknown dynamics:
\begin{equation}
\dot{\mathbf{x}}(t)
=
A\mathbf{x}(t)
+
B\mathbf{u}(t),
\label{eq:LTI_example}
\end{equation}
where \(A \in \mathbb{R}^{n \times n}\) is a constant system matrix and \(B \in \mathbb{R}^{n \times m}\) is a known input matrix. For the system \eqref{eq:LTI_example} we consider the following nominal state-feedback controller
\begin{equation}
\mathbf{u}^\star(t) = K_x\mathbf{x}(t) + K_r\mathbf{r}(t),
\label{eq:nominal_controller_example}
\end{equation}
where $K_x$ denotes the feedback gain and $K_r$ the feedforward gain associated with the reference input $r(t)$. The corresponding nominal closed-loop dynamics are defined as
\begin{equation*}
\dot{\mathbf{x}}(t)
=
A_{\mathrm{cl}}\mathbf{x}(t)
+
B_{\mathrm{cl}}\mathbf{r}(t),
\end{equation*}
with
\begin{align*}
A_{\mathrm{cl}} &= A + BK_x,\\
B_{\mathrm{cl}} &= BK_r.
\end{align*}

The local closed-loop dynamics are governed by the eigenstructure of~$A_{\mathrm{cl}}$, with the eigenvalues determining the temporal evolution and the eigenvectors defining the associated modal directions.

As discussed in Section \ref{sec:Preliminaries}, within the standard paradigm for CBF-based state constraint enforcement, safety for \eqref{eq:LTI_example} is imposed at the control input-level via the quadratic program
\begin{equation}
\begin{aligned}
\mathbf{u}(t) = \arg\min_{\mathbf{u}(t)\in U} \quad
& \|\mathbf{u}(t) - \mathbf{u}^\star(t)\|_2^2 \\
\text{s.t.}\quad
& a(\mathbf{x}(t))\,\mathbf{u}(t) + b(\mathbf{x}(t)) \ge 0,
\end{aligned}
\label{eq:qp_cbf_example}
\end{equation}
where
\begin{equation*}
a(\mathbf{x}(t)) = \frac{\partial h}{\partial \mathbf{x}}B,
\qquad
b(\mathbf{x}(t)) = \frac{\partial h}{\partial \mathbf{x}}A\mathbf{x}(t) + \alpha\!\left(h(\mathbf{x}(t))\right).
\end{equation*}

The corresponding Lagrangian associated with \eqref{eq:qp_cbf_example} is given by
\begin{equation*}
\mathcal{L}(\mathbf{u}(t),\lambda)
=
\frac{1}{2}\|\mathbf{u}(t)-\mathbf{u}(t)^\star\|_2^2
-
\lambda\big(a(\mathbf{x}(t))\,\mathbf{u}(t)+b(\mathbf{x}(t))\big),
\end{equation*}
where $\lambda(\mathbf{x}(t)) \ge 0$ denotes the Lagrange multiplier associated with the safety constraint. The explicit solution admits the closed-form expression
\begin{equation}
\mathbf{u}(t)
=
\mathbf{u}(t)^\star(\mathbf{x}(t))
+
\lambda(\mathbf{x}(t))\,a(\mathbf{x}(t))^\top,
\label{eqn:lagrange_depending_controller}
\end{equation}
where the Lagrange multiplier is given by
\begin{equation*}
\lambda(\mathbf{x}(t))
=
\max\!\left\{0,\,
-\frac{\delta(\mathbf{x}(t))}{a(\mathbf{x}(t))a(\mathbf{x}(t))^\top}
\right\},
\end{equation*}
and
\begin{equation*}
\delta(\mathbf{x}(t)) := a(\mathbf{x}(t))\mathbf{u}^\star(t)(\mathbf{x}(t)) + b(\mathbf{x}(t)).
\end{equation*}

Substituting the nominal controller from \eqref{eq:nominal_controller_example} yields
\begin{equation*}
\mathbf{u}(t)
=
K_x\mathbf{x}(t)
+
K_r\mathbf{r}(t)
+
\lambda(\mathbf{x}(t))\,a(\mathbf{x}(t))^\top.
\end{equation*}
Hence, the resulting closed-loop dynamics are
\begin{equation*}
\dot{\mathbf{x}}(t)
=
A_{\mathrm{cl}}\mathbf{x}(t)
+
B_{\mathrm{cl}}\mathbf{r}(t)
+
B\lambda(\mathbf{x})a(\mathbf{x})^\top.
\label{eqn:constraint_depending_closed_Loop}
\end{equation*}

This expression shows directly that whenever the safety constraint is active, i.e., when
\[
\delta(\mathbf{x}(t)) < 0 \quad \Longleftrightarrow \quad \lambda(\mathbf{x}(t)) \neq 0,
\]
the closed-loop vector field is modified by the additional state-dependent term
\[
B\lambda(\mathbf{x}(t))a(\mathbf{x}(t))^\top.
\]

Consequently, the system is no longer governed by the nominal closed-loop matrix $A_{\mathrm{cl}}$. Instead, around an operating point $\mathbf{x}_e$ where the active set remains unchanged and $\lambda(\mathbf{x})$ is locally smooth, the effective local closed-loop Jacobian becomes
\begin{equation}
A_{\mathrm{eff}}(\sigma_0,\mathbf{x}_e)
=
A_{\mathrm{cl}}
+
B
\frac{\partial}{\partial \mathbf{x}}
\big(\lambda(\mathbf{x}(t))a(\mathbf{x}(t))^\top\big)
\Big|_{\mathbf{x}=\mathbf{x}_e}.
\label{eqn:constraint_depending_system_matrix}
\end{equation}

Therefore, the local eigenstructure in regions where the constraint is active is determined by $A_{\mathrm{eff}}(\mathbf{x}_e)$ rather than by the nominal matrix $A_{\mathrm{cl}}$. In contrast, whenever the constraint is inactive ($\lambda(\mathbf{x})=0$), the nominal closed-loop dynamics, and thus the eigenstructure, are exactly preserved.

\medskip

To illustrate this effect, we consider a simple second-order system with double-integrator structure in the form of \eqref{eq:LTI_example}. The state vector is given by $\mathbf{x}(t) = [x_1(t) \;\; x_2(t)]^\top$ and the scalar control input $u$. This model captures an essential second-order behavior and enables a transparent analysis of the interaction between nominal control and safety filtering in the considered application domain. The system matrices are given by
\begin{equation*}
A =
\begin{bmatrix}
0.0 & 1.0\\
25.0 & 2.0
\end{bmatrix}, \qquad
B =
\begin{bmatrix}
0.0\\
20.0
\end{bmatrix}.
\end{equation*}
The nominal controller follows the structure introduced in \eqref{eq:nominal_controller_example} and is designed via pole placement such that the closed-loop poles are located at $\{-6 \pm 3i\}$.

\medskip

For the considered motivating example, safety is imposed via bounds on $x_2(t)$, i.e.,
\begin{equation}
x_{2,\min} \leq x_2(t) \leq x_{2,\max},
\end{equation}
where $x_{2,\min} = -30$ and $x_{2,\max} = 30$. This is enforced through the CBF candidates
\begin{equation*}
h_1(\mathbf{x}(t)) = x_{2,\max} - x_2(t),
\end{equation*}
\begin{equation*}
h_2(\mathbf{x}(t)) = x_2(t) - x_{2,\min},
\end{equation*}
which, together with the OLM-CBF-based safety filter in \eqref{eq:qp_cbf_example}, guarantee forward invariance of the admissible set.

\medskip

\begin{figure}
    \centering
    \includegraphics[width=\columnwidth]{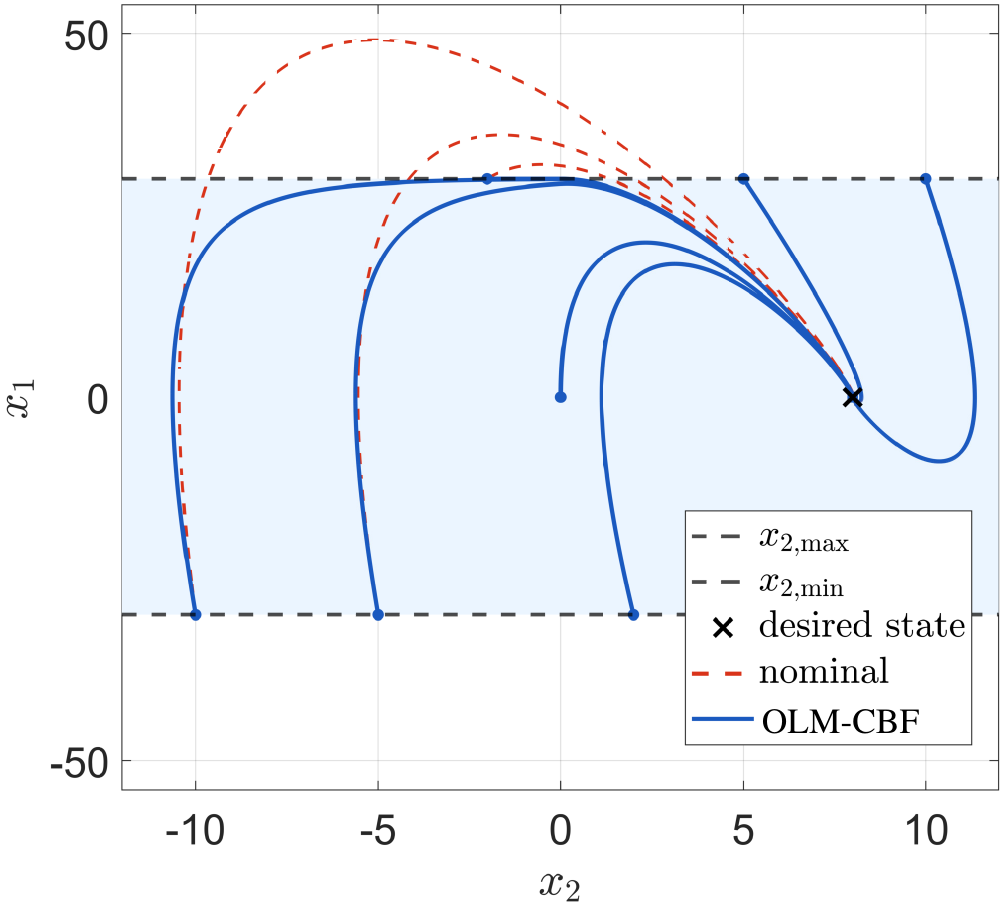}
    \caption{Closed-loop trajectories for multiple initial states with and without an OLM-CBF-based safety filter.}
    \label{fig:phase_plot}
\end{figure}

Fig.~\ref{fig:phase_plot} shows the resulting closed-loop trajectories for different initial states, both comfortably inside the admissible set and on its boundary. The control objective is to steer the system toward the desired equilibrium $\mathbf{x} = [8 , 0]^\top$, which lies well within the safe set. The nominal controller without safety filtering (red dashed) achieves the control objective but violates the pitch rate constraints, as trajectories exceed the prescribed bounds. In contrast, the OLM-CBF-based safety filter (blue) modifies the control input such that all trajectories remain within the admissible set while still converging to the desired equilibrium. These results demonstrate that the OLM-CBF framework effectively enforces safety for the considered scenario. However, as shown in the following, this safety enforcement alters the closed-loop dynamics, which comes at the expense of losing guarantees on robustness margins.

\medskip

To analyze this effect, the system is evaluated at representative states within the safe set. In particular, we distinguish between regions where the safety constraint is inactive and regions where it becomes active, as characterized by the multiplier $\lambda(\mathbf{x}(t))$ in the closed-form solution of the quadratic program. As discussed earlier, the nominal controller is designed such that the closed-loop poles are located at $\{-6 \pm 3i\}$. This design yields gain margins of approximately $8.14$ dB and phase margins of approximately $47.2^\circ$, which are consistent with typical robustness requirements in flight control applications.

\begin{figure}
    \centering
    \includegraphics[width=\columnwidth]{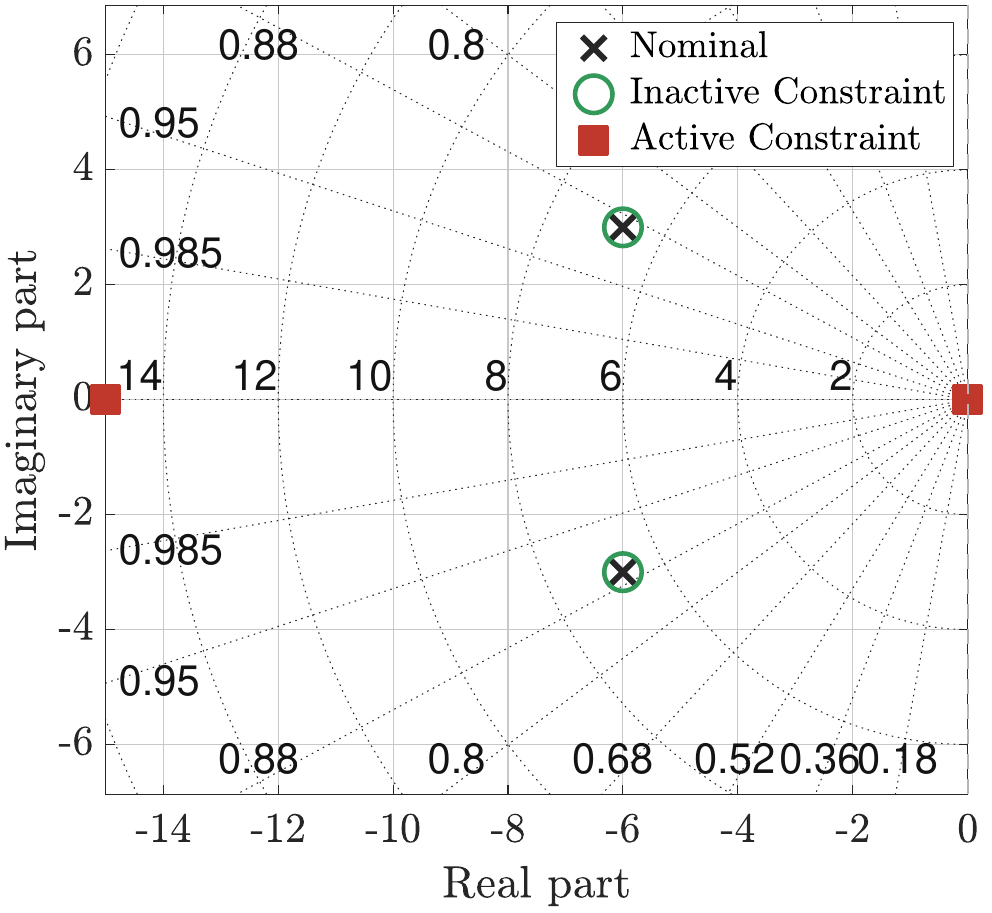}
    \caption{Closed-loop pole locations for nominal, inactive, and active safety filter cases.}
    \label{fig:poles}
\end{figure}

For states well inside the safe set, it can be assumed that the safety constraint remains inactive, i.e., $\lambda(\mathbf{x}(t)) = 0$. In this case, the safety filter does not modify the control input, and the resulting closed-loop dynamics should coincide with the nominal system. This is confirmed by the pole locations shown in Fig.~\ref{fig:poles}, where the inactive case matches the nominal poles. In contrast, for states near the boundary of the safe set, the constraint becomes active, i.e., $\lambda(\mathbf{x}(t)) > 0$. As a result, the control input is modified, leading to a state-dependent alteration of the closed-loop dynamics. As illustrated in Fig.~\ref{fig:poles}, this modification leads to a pronounced shift in the pole locations. While one pole moves further into the left half-plane, the second pole shifts toward the imaginary axis.

\begin{table*}[t]
\centering
\caption{Comparison of closed-loop pole locations and robustness margins for nominal, inactive, and active safety filter cases.}
\label{tab:poles_margins}
\begin{tabular}{lccc}
\toprule
\textbf{Case} & \textbf{Poles} & \textbf{Gain Margin (dB)} & \textbf{Phase Margin (deg)} \\
\midrule
Nominal 
& $-6 \pm 3i$ 
& $\pm 8.14$ 
& $\pm 47.22$ \\

Inactive OLM-CBF 
& $-6 \pm 3i$ 
& $\pm 8.14$ 
& $\pm 47.22$ \\

Active OLM-CBF 
& $0,\,-15$ 
& $0.00$ 
& $0.00$ \\
\bottomrule
\end{tabular}
\end{table*}

The quantitative results summarized in Table~\ref{tab:poles_margins} further clarify this behavior. While the nominal and inactive cases exhibit identical pole locations and robustness margins, the active case shows a complete loss of gain and phase margins. Although the shift in pole locations alone does not necessarily imply problematic behavior, the associated loss of robustness is critical. In particular, the phase margin approaches $0^\circ$ and the gain margin approaches $0$ dB, indicating operation at the boundary of stability and a high sensitivity to model uncertainties and disturbances.


These results demonstrate that, although the OLM-CBF-based safety filter guarantees constraint satisfaction, it can significantly alter the closed-loop dynamics when active. In particular, robustness properties ensured by the nominal controller design are not preserved in general. This highlights a fundamental limitation of input-level safety filtering, as safety enforcement is achieved without explicitly accounting for the resulting closed-loop behavior.


%% file: Sections/Closed_Loop_Model_CBFs.tex
Classical CBF-based approaches enforce safety by directly modifying the control input $u(t)$ through a QP-based safety filter. While this provides formal safety guarantees, such input-level filtering can interfere with the closed-loop behavior of the system. In particular, it may degrade stability properties, induce undesirable interactions with internal controller states (e.g., integral windup), and complicate integration into existing flight control architectures. These limitations are especially critical in aerospace applications, where control systems are carefully tuned and certified, and modifications to the control input can necessitate extensive redesign and validation. To address these challenges, this work adopts a closed-loop model-based perspective. In the following, this concept is formalized, and a CBF-based safety filter operating on the reference signal is derived using closed-loop model information.

\subsection{Robust Closed-Loop Controller Design}
\label{sec:CL_Controller_Design}

Building upon the linearized short-period model introduced in \eqref{eq:LTI_short_period}, a feedback controller is designed to ensure robust stability with required performance specifications. The objective is to obtain a well-behaved closed-loop system that satisfies robustness and performance criteria, while providing a reliable basis for subsequent safety filtering at the reference-level. To this end, a feedback control law of the form
\begin{equation}
u(t) = K_x(\sigma_0) \mathbf{x}(t) + K_r(\sigma_0) \mathbf{r}(t)
\label{eqn:robust_controll_law}
\end{equation}
is considered, where $K_x(\sigma_0)$ denotes the feedback gain and $K_r(\sigma_0)$ the feedforward gain associated with the reference input $r(t)$ around an operating point $\sigma_0$. The controller is designed using a multi-objective robust control framework, combining frequency-domain robustness requirements and performance specifications. In particular, the design enforces robustness margins, including prescribed gain and phase margins, to ensure sufficient tolerance against modeling inaccuracies and unmodeled dynamics. 

We define the following transfer functions 
\begin{align}
e(s)&=N_1(s)r(s),\notag\\
u(s)&=N_2(s)r(s),\notag\\
x(s)&=N_3(s)r(s),\notag
\end{align}
where 
\begin{align}
e(s)&=\mathcal{L}\{e(t)\}, 
\qquad u(s)=\mathcal{L}\{u(t)\}, \notag\\
x(s)&=\mathcal{L}\{x(t)\},
\qquad r(s)=\mathcal{L}\{r(t)\} \notag
\end{align}
are the Laplace transformed signals of the tracking error, the control input, the state vector and the reference signal, respectively.
Further, we introduce the signal-based $\mathcal{H}_\infty$ performance criteria
\begin{equation}
J(K_x(\sigma_0),K_r(\sigma_0)) 
=
\|N(K_x(\sigma_0),K_r(\sigma_0))\|_\infty.
\end{equation}
with
\begin{align}
N(s)
=
\begin{bmatrix}
W_1(s)N_1(s) & 0 & 0\\
0 & W_2(s)N_2(s) & 0\\
0 & 0 & W_3(s)I_{2\times2}N_3(s)\\
\end{bmatrix}\notag
\end{align}
and scalar weightings $W_1(s), W_2(s), W_3(s)\in\mathcal{R}(s)$ selected from the set of proper rational transfer functions $\mathcal{R}(s)$. 
By designing the weighting function $W_1(s)$, we can define the desired controller bandwidth, within which accurate tracking can be achieved. To limit the amount of control action and avoid excessively high gains, we choose $W_2(s)$ accordingly. Finally, to avoid high-frequency responses of the system and to attenuate noise we design $W_3(s)$.
To make the $\mathcal{H}_\infty$ control design robust, we formulate the following constrained optimization problem
\begin{equation}
\begin{aligned}
J^{*}=\min_{K_x,\,K_r} \quad & J(K_x(\sigma_0),K_r(\sigma_0)) \\
\text{s.t.} \quad 
& \mathrm{GM} \geq \mathrm{GM}_{\min}, \\
& \mathrm{PM} \geq \mathrm{PM}_{\min},
\label{eq:multi_objective_optimization}
\end{aligned}
\end{equation}
which enforces the required minimum gain and phase margins, $\mathrm{GM}_{\min}$ and $\mathrm{PM}_{\min}$, respectively. In general, the optimization problem (\ref{eq:multi_objective_optimization}) is nonlinear, non-convex and non-smooth but can be recast into a multi-objective optimization problem and solved using non-smooth programming techniques \cite{apkarian2014multi}.
\begin{figure*}[!tbp]
    \centering
    \begin{subfigure}[b]{0.35\textwidth}
        \centering
        \includegraphics[width=1\textwidth]{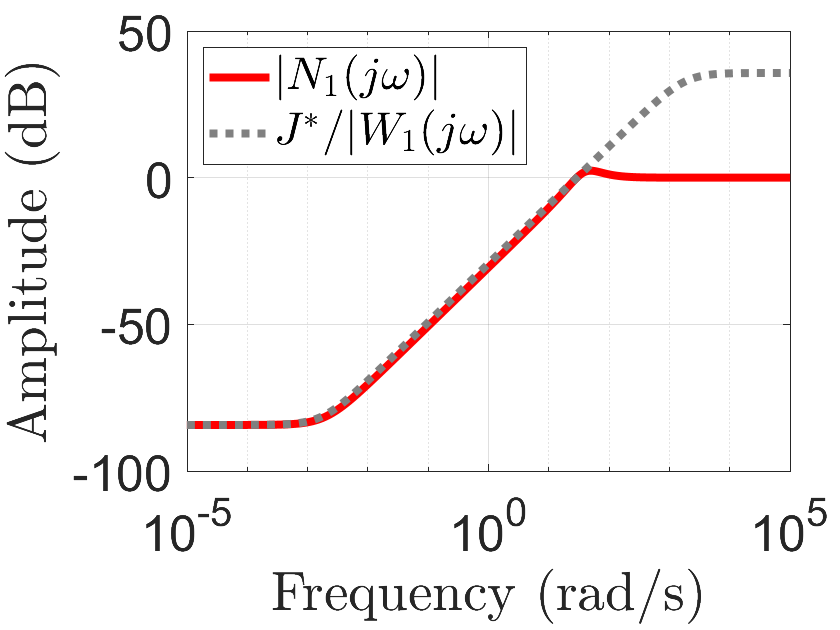}
        \caption{Error sensitivity\label{fig:robust_controller_evaluation_a}}
    \end{subfigure}
    \begin{subfigure}[b]{0.35\textwidth}
        \centering
        \includegraphics[width=1\textwidth]{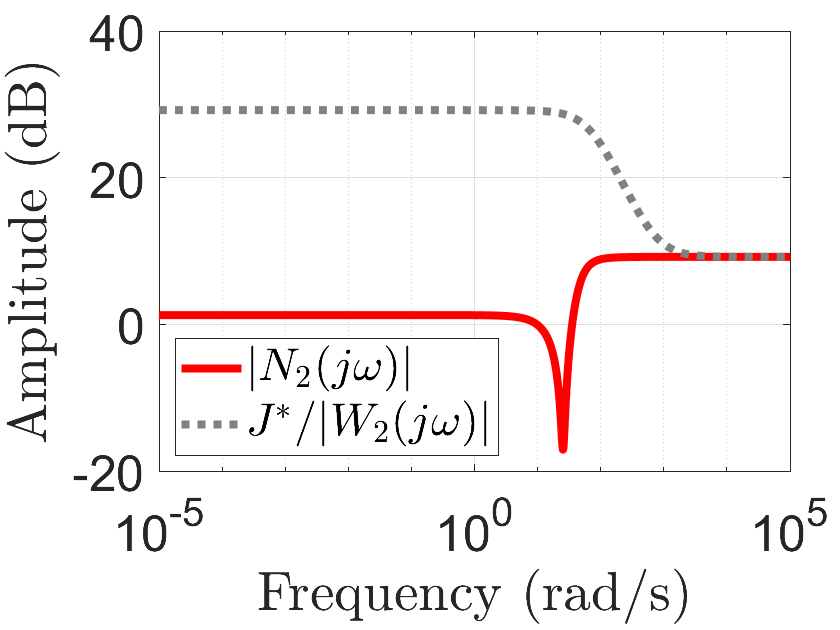}
        \caption{Control effort\label{fig:robust_controller_evaluation_b}}
    \end{subfigure}\\
    \begin{subfigure}[b]{0.35\textwidth}
        \centering
        \includegraphics[width=1\textwidth]{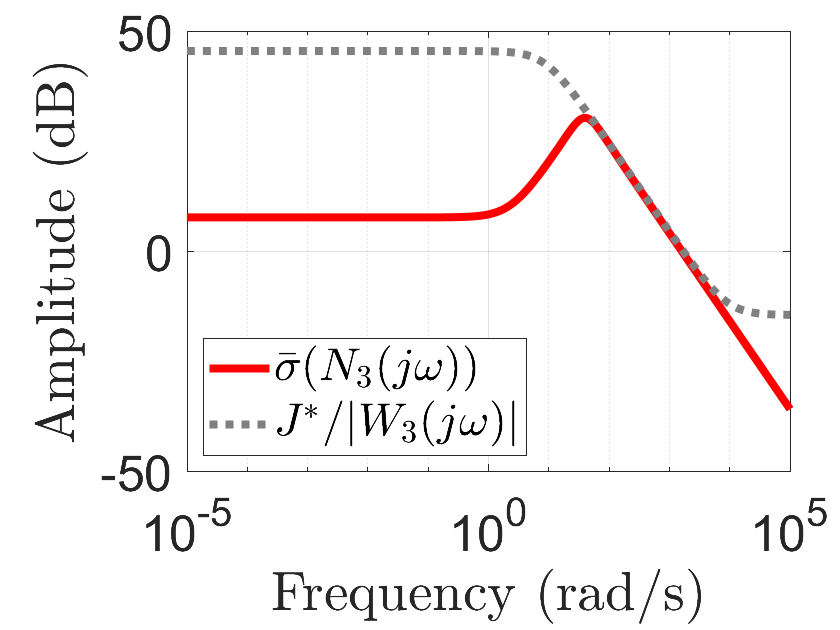}
        \caption{High frequency attenuation\label{fig:robust_controller_evaluation_c}}
    \end{subfigure}
    \begin{subfigure}[b]{0.35\textwidth}
        \centering
        \includegraphics[width=1\textwidth]{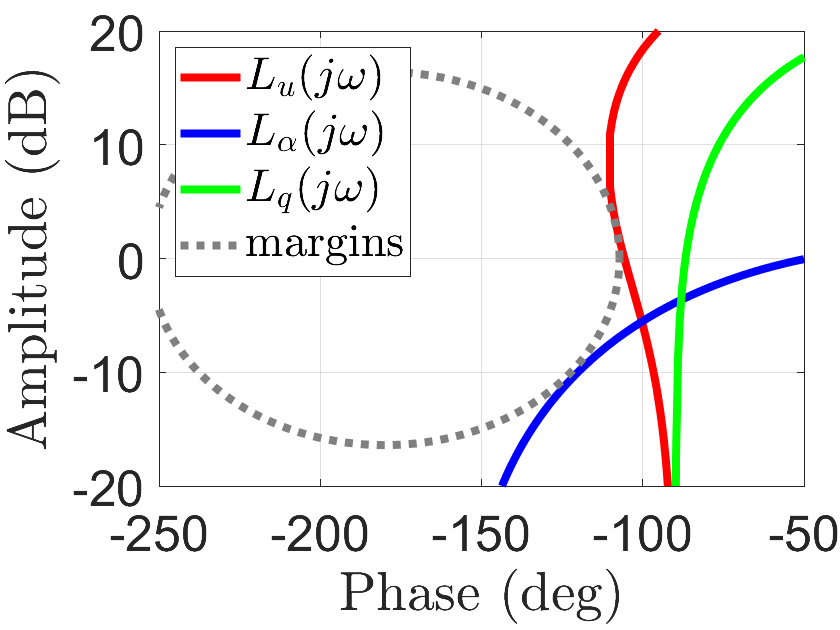}
        \caption{Disk margins\label{fig:robust_controller_evaluation_d}}
    \end{subfigure}
	\caption{Performance evaluation of robust controller.} 
	\label{fig:robust_controller_evaluation}
\end{figure*}

The results of the robust controller synthesis are visualized in Fig. \ref{fig:robust_controller_evaluation}. The desired performance criteria defined by the scaled weightings $J^{*}/|W_i(j\omega)|$ are met by the synthesized closed-loop frequency responses $N_i(j\omega)$. In particular, the achieved performance with respect to tracking error sensitivity, control effort reduction and noise attenuation are presented by Figs. \ref{fig:robust_controller_evaluation}a-c. Further, we consider the input broken-loop frequency response $L_u(j\omega)$ and the output broken-loop frequency responses $L_\alpha(j\omega)$ and $L_q(j\omega)$ to verify the robustness of the synthesized controller.  In Fig. \ref{fig:robust_controller_evaluation}d, the robustness with respect to simultaneous phase and gain variations, known as disk margins, is visualized. For all open-loop responses we achieve strong robustness with a gain margin of at least $\SI{16}{dB}$ and a phase margin of at least $\SI{70}{\deg}$. The obtained control bandwidth is $\SI{21.21}{\radian/\s}$.

Substituting \eqref{eqn:robust_controll_law} into \eqref{eq:LTI_short_period} yields
\begin{equation}
\dot{\mathbf{x}}(t)
=
A_{\mathrm{cl}}(\sigma_0) \mathbf{x}(t)
+
B_{\mathrm{cl}}(\sigma_0) \mathbf{r}(t),
\label{eq:closed_loop_system_filter}
\end{equation}
with
\begin{align}
A_{\mathrm{cl}}(\sigma_0) &= A(\sigma_0) + B(\sigma_0)K_x(\sigma_0),\notag\\
B_{\mathrm{cl}}(\sigma_0) &= B(\sigma_0)K_r(\sigma_0).\notag
\end{align}

This closed-loop model explicitly characterizes the mapping from the commanded reference signal to the system state and forms the foundation for the subsequent development of the reference-level safety filter. In particular, it enables the safety constraints to be formulated with respect to the actual closed-loop dynamics, thereby ensuring that the desired stability and performance properties are preserved, whenever it is possible during the safety-critical operation.

\subsection{Closed-Loop Model-Based Safety Filter}
\label{sec:CL_Safety_Filter}

\begin{figure*}
    \centering
    \includegraphics[width=0.8\textwidth]{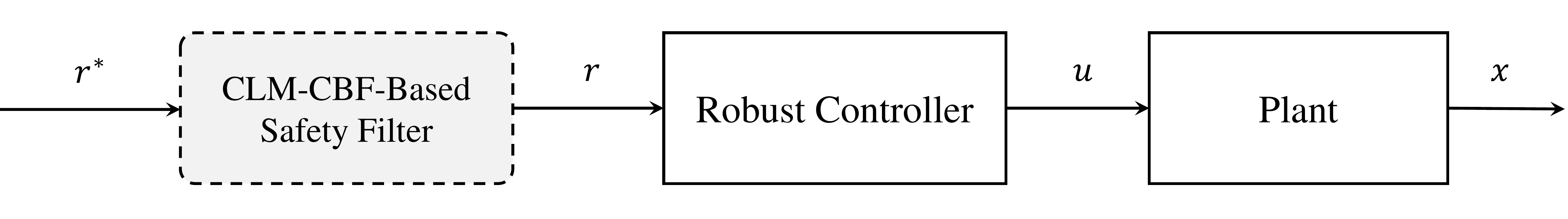}
    \caption{Conceptual control architecture with a closed-loop model-based Control Barrier Function (CLM-CBF).}
    \label{fig:CLM-CBF}
\end{figure*}

Based on the closed-loop controller design from the previous subsection, safety is now enforced directly with respect to the closed-loop system dynamics. As shown in Fig.~\ref{fig:CLM-CBF}, the key idea is to shift the barrier-function-based safety condition from the control input-level to the reference-level by explicitly exploiting the mapping from the commanded reference signal to the resulting closed-loop state evolution. In contrast to classical input-level CBF formulations, the present approach does not modify the control law itself, but instead computes a certified reference signal such that the resulting closed-loop trajectory remains within the admissible safe set.


Let the operational safe set be denoted by
\begin{equation}
\label{eq:cl_safe_set}
\mathcal{S}
=
\Bigl\{
\mathbf{x}\in\mathbb{R}^n
\mid
h(\mathbf{x})\geq 0
\Bigr\}.
\end{equation}


The objective is to guarantee forward invariance of \(\mathcal{S}\) for the closed-loop system in \eqref{eq:closed_loop_system_filter}. The following result establishes the barrier condition for the closed-loop system.

\begin{theorem}
\label{thm:robust_cl_cbf}
Let the safe set \(S\) be defined by \eqref{eq:cl_safe_set}, and consider the closed-loop dynamics \eqref{eq:closed_loop_system_filter}. Assume that there exists an extended class-\(\mathcal{K}\) function \(\alpha(\cdot)\) such that, for all \(x \in S\), the reference input \(r\) satisfies
\begin{equation}
\frac{\partial h(\mathbf{x})}{\partial \mathbf{x}}\Big(A_{\mathrm{cl}}(\sigma_0) \mathbf{x}(t) + B_{\mathrm{cl}}(\sigma_0) \mathbf{r}(t)\Big)
\geq
-\alpha\big(h(\mathbf{x}(t))\big).
\label{eq:cl_cbf_condition}
\end{equation}
Then the set \(S\) is forward invariant under the closed-loop dynamics \eqref{eq:closed_loop_system_filter}.
\end{theorem}

\begin{proof}
The proof follows the standard CBF argument applied to the closed-loop system. In particular, on the boundary of the safe set, $\partial S \triangleq \{\mathbf{x}(t) \in \mathbb{R}^n \mid h(\mathbf{x}(t))=0\}$, condition \eqref{eq:cl_cbf_condition} reduces to $\dot{h}(\mathbf{x}(t)) = \frac{\partial h(\mathbf{x})}{\partial \mathbf{x}}\Big(A_{\mathrm{cl}}(\sigma_0) \mathbf{x}(t) + B_{\mathrm{cl}}(\sigma_0) \mathbf{r}(t)\Big)
\geq 0$. Thus, the closed-loop vector field points inward or tangentially on the boundary of \(S\). By Nagumo's theorem, this is sufficient to guarantee forward invariance of \(S\) \cite{Nagumo_1942}. Hence, every trajectory starting in \(S\) remains in \(S\) for $t\geq0$.
\end{proof}

Theorem~\ref{thm:robust_cl_cbf} shows that the closed-loop model can be directly used for safety certification within the reference-level safety filter. This result formalizes the central idea of the proposed approach: safety is enforced with respect to the closed-loop system dynamics, rather than the open-loop plant. As a consequence, the safety filter does not modify the control law itself, but instead adjusts the commanded reference signal such that the resulting closed-loop trajectory remains within the admissible set.

Let \(\mathbf{r}^\star(t)\) denote the desired, potentially unsafe, reference command. The safety filter computes the closest admissible reference signal \(\mathbf{r}(t)\) that satisfies the closed-loop barrier condition. This leads to the quadratic program

\begin{align}
\mathbf{r}(t)&=\arg\min_{\mathbf{r}(t)\in\mathbb{R}^p}\quad \|\mathbf{r}(t)-\mathbf{r}^\star(t)\|_W^2 \label{eq:cl_reference_filter_qp}\\
&\text{s.t.}\notag\\ 
&\frac{\partial h(\mathbf{x}(t))}{\partial \mathbf{x}}\Big(A_{\mathrm{cl}}(\sigma_0) \mathbf{x}(t)+B_{\mathrm{cl}}(\sigma_0) \mathbf{r}(t)\Big)
\geq -\alpha\!\big(h(\mathbf{x}(t))\big),\notag
\end{align}
with \( \lVert \cdot \rVert_W^2 := (\cdot)^\top W (\cdot) \) denoting a weighted \(L_2\)-norm.

The optimization problem \eqref{eq:cl_reference_filter_qp} computes the least-invasive modification of the commanded reference signal while ensuring forward invariance of the safe set with respect to the closed-loop dynamics. Since the optimization is performed at the reference-level, the feedback controller remains unchanged and retains its original stabilization and tracking properties.

\begin{remark}[Connection to reference governors]
\label{rem:reference_governance}
The proposed CLM-CBF framework is closely related to the concept of reference governors, wherein the commanded reference signal is modified to ensure constraint satisfaction of the closed-loop system. In this sense, the proposed approach can be interpreted as a control barrier function-based realization of a reference governor, in which safety is enforced through inequality constraints derived from barrier functions. 
For comprehensive overviews of reference governor techniques, the reader is referred to \cite{Garone2017,Kolmanovsky2014}.
\end{remark}

\section{Pole-Region Preservation Analysis for CLM-CBF Safety Filters}

The preceding subsection introduced the concept of CLM-CBF safety filters as a reference-level supervisory mechanism for enforcing forward invariance of the safe set. In contrast to input-level CBF-QP formulations, the proposed safety filter does not directly modify the control input generated by the feedback law in \eqref{eqn:robust_controll_law}. 
This architectural separation is expected to be beneficial from a certification perspective, as the safety-critical supervision is introduced at the reference level while preserving the structure of the underlying flight control law.

However, similar to the discussion in Section~\ref{sec:Motivating Example} for OLM-CBF concepts, preservation of the controller realization and safety enforcement at the reference level does not per se imply unconditional preservation of the closed-loop pole locations during active safety intervention. The reason is that, despite the higher-level intervention mechanism, the certified reference signal computed by the CLM-CBF QP is, in general, a state-dependent optimizer map. Let this map be denoted by
\begin{equation}
\mathbf{r}(t)
=
\pi\!\left(
\mathbf{x}(t),
\mathbf{r}^{\star}(t)
\right),
\label{eq:clm_optimizer_map}
\end{equation}
where \(\mathbf{r}^{\star}(t)\) denotes the desired reference command and \(\pi(\cdot,\cdot)\) is induced by the optimization problem \eqref{eq:cl_reference_filter_qp}.

\begin{remark}
To improve readability, the dependence on the operating point \(\sigma_0\) is suppressed throughout this subsection and restored in the subsequent parts of the manuscript. Thus, we write \(A_{\mathrm{cl}}=A_{\mathrm{cl}}(\sigma_0)\) and \(B_{\mathrm{cl}}=B_{\mathrm{cl}}(\sigma_0)\), with the dependence on \(\sigma_0\) understood from context.
\end{remark}

Substitution of \eqref{eq:clm_optimizer_map} into the closed-loop model \eqref{eq:closed_loop_system_filter} gives the exact safety-filtered closed-loop vector field
\begin{equation}
\dot{\mathbf{x}}(t)
=
\mathbf{f}_{\mathrm{sf}}
\!\left(
\mathbf{x}(t),
\mathbf{r}^{\star}(t)
\right)
=
A_{\mathrm{cl}}\mathbf{x}(t)
+
B_{\mathrm{cl}}
\pi\!\left(
\mathbf{x}(t),
\mathbf{r}^{\star}(t)
\right).
\label{eq:clm_cbf_closed_loop_map}
\end{equation}
Thus, when the safety filter is active, the CLM-CBF introduces an additional supervisory feedback path through the reference channel. Consequently, the complete safety-filtered system must be distinguished from the nominal closed-loop system driven by an exogenous reference command.

To make this effect explicit, consider an operating point \((\mathbf{x}_e,\mathbf{r}^{\star}_e)\) satisfying
\begin{equation}
\mathbf{f}_{\mathrm{sf}}
\!\left(
\mathbf{x}_e,
\mathbf{r}^{\star}_e
\right)
=
0.
\label{eq:safety_filtered_equilibrium}
\end{equation}
Define the local perturbation variables
\begin{equation}
\delta\mathbf{x}
=
\mathbf{x}-\mathbf{x}_e,
\qquad
\delta\mathbf{r}^{\star}
=
\mathbf{r}^{\star}-\mathbf{r}^{\star}_e .
\label{eq:local_perturbation_variables}
\end{equation}
Assuming that the optimizer map is differentiable at \((\mathbf{x}_e,\mathbf{r}^{\star}_e)\), the first-order variation of the certified reference is
\begin{equation}
\delta \mathbf{r}
=
\left.
\frac{\partial \pi}{\partial \mathbf{x}}
\right|_{(\mathbf{x}_e,\mathbf{r}^{\star}_e)}
\delta \mathbf{x}
+
\left.
\frac{\partial \pi}{\partial \mathbf{r}^{\star}}
\right|_{(\mathbf{x}_e,\mathbf{r}^{\star}_e)}
\delta \mathbf{r}^{\star}
+
\mathcal{O}
\!\left(
\left\|
\begin{bmatrix}
\delta\mathbf{x}\\
\delta\mathbf{r}^{\star}
\end{bmatrix}
\right\|^2
\right).
\label{eq:optimizer_map_expansion}
\end{equation}
The corresponding linearized safety-filtered dynamics are
\begin{equation}
\delta\dot{\mathbf{x}}
=
A_{\mathrm{eff}}\delta\mathbf{x}
+
B_{\mathrm{eff}}\delta\mathbf{r}^{\star},
\label{eq:linearized_safety_filtered_dynamics}
\end{equation}
where
\begin{equation}
A_{\mathrm{eff}}
=
A_{\mathrm{cl}}
+
\Delta A_{\pi},
\label{eq:effective_closed_loop_matrix_clm}
\end{equation}
with
\begin{equation}
\Delta A_{\pi}
=
B_{\mathrm{cl}}
\left.
\frac{\partial \pi}{\partial \mathbf{x}}
\right|_{(\mathbf{x}_e,\mathbf{r}^{\star}_e)},
\label{eq:optimizer_induced_perturbation}
\end{equation}
and
\begin{equation}
B_{\mathrm{eff}}
=
B_{\mathrm{cl}}
\left.
\frac{\partial \pi}{\partial \mathbf{r}^{\star}}
\right|_{(\mathbf{x}_e,\mathbf{r}^{\star}_e)} .
\label{eq:effective_reference_matrix}
\end{equation}
Equation~\eqref{eq:effective_closed_loop_matrix_clm} shows that active CLM-CBF intervention may shift the local closed-loop poles through the optimizer sensitivity \(\partial \pi/\partial \mathbf{x}\). Exact pole preservation during active safety filtering would require
\begin{equation}
B_{\mathrm{cl}}
\frac{\partial \pi}{\partial \mathbf{x}}
=
0,
\label{eq:exact_pole_preservation_condition}
\end{equation}
which is generally not satisfied for a state-dependent safety filter. Therefore, the appropriate objective is not exact pole preservation, but preservation of a prescribed pole region associated with the required robustness margins, e.g., those used during the robust controller synthesis in Section~\ref{sec:CL_Controller_Design}.

It is important to emphasize that the pole-region analysis is inherently a local linear concept. The exact safety-filtered dynamics \eqref{eq:clm_cbf_closed_loop_map} may remain nonlinear due to the state dependence of the optimizer map. Therefore, the nonlinear terms not captured by the first-order approximation are explicitly retained in the following. Define the nonlinear residual
\begin{align}
\mathbf{d}
\!\left(
\delta\mathbf{x},
\delta\mathbf{r}^{\star}
\right)
&=
\mathbf{f}_{\mathrm{sf}}
\!\left(
\mathbf{x}_e+\delta\mathbf{x},
\mathbf{r}^{\star}_e+\delta\mathbf{r}^{\star}
\right)
\nonumber\\
&\quad
-
A_{\mathrm{eff}}\delta\mathbf{x}
-
B_{\mathrm{eff}}\delta\mathbf{r}^{\star}.
\label{eq:nonlinear_residual_definition}
\end{align}
Since \((\mathbf{x}_e,\mathbf{r}^{\star}_e)\) satisfies \eqref{eq:safety_filtered_equilibrium}, it follows that
\begin{equation}
\mathbf{d}(\mathbf{0},\mathbf{0})=\mathbf{0}.
\end{equation}
The exact nonlinear safety-filtered dynamics can therefore be written as
\begin{equation}
\delta\dot{\mathbf{x}}
=
A_{\mathrm{eff}}\delta\mathbf{x}
+
B_{\mathrm{eff}}\delta\mathbf{r}^{\star}
+
\mathbf{d}
\!\left(
\delta\mathbf{x},
\delta\mathbf{r}^{\star}
\right).
\label{eq:exact_local_safety_filtered_dynamics}
\end{equation}
Thus, the first-order approximation is not used to discard the nonlinearities. Rather, it separates the local pole-forming dynamics from a higher-order residual term, which can be bounded on a neighborhood of the operating point.

Let the desired pole region be defined as the shifted open left half-plane
\begin{equation}
\mathcal{D}_{\sigma}
=
\Bigl\{
\lambda\in\mathbb{C}
:
\operatorname{Re}(\lambda)<-\sigma
\Bigr\},
\qquad
\sigma>0,
\label{eq:desired_pole_region}
\end{equation}
where \(\sigma\) specifies the required minimum exponential decay rate. The following analysis assumes that the nominal closed-loop matrix obtained from the robust controller synthesis admits a shifted Lyapunov certificate for this prescribed decay rate. That is, there exist matrices \(P=P^T>0\) and \(Q=Q^T>0\) such that
\begin{equation}
A_{\mathrm{cl}}^T P
+
P A_{\mathrm{cl}}
+
2\sigma P
\leq
-Q.
\label{eq:nominal_pole_region_certificate}
\end{equation}
This condition is equivalent to
\begin{equation}
(A_{\mathrm{cl}}+\sigma I)^T P
+
P(A_{\mathrm{cl}}+\sigma I)
\leq
-Q,
\end{equation}
and is feasible whenever \(A_{\mathrm{cl}}+\sigma I\) is Hurwitz. Thus, \eqref{eq:nominal_pole_region_certificate} is not imposed for an arbitrary \(\sigma\), but is a verifiable certificate that the synthesized nominal closed-loop system has a decay margin larger than \(\sigma\). In practice, \(\sigma\) is selected below the nominal closed-loop decay rate achieved by the robust controller, and the matrices \(P\) and \(Q\) are obtained by solving the corresponding Lyapunov inequality. The matrix \(Q\), together with \(P\), then defines a computable Lyapunov decay reserve that can be used to bound admissible perturbations of the closed-loop system matrix.

The effect of the CLM-CBF safety filter can now be bounded through the sensitivity of the optimizer map. Let \(\Omega\subset\mathbb{R}^n\) denote the operating region in which the safety filter is expected to operate, and define
\begin{equation}
L_{\pi}
=
\sup_{\mathbf{x}\in\Omega}
\left\|
\frac{\partial \pi}{\partial \mathbf{x}}
\right\|.
\label{eq:optimizer_sensitivity_bound}
\end{equation}
The quantity \(L_{\pi}\) measures the maximum local gain from state variations to reference modifications induced by the CLM-CBF optimizer. It therefore characterizes the aggressiveness of the supervisory safety layer.

The following theorem gives a sufficient condition under which the linearized active CLM-CBF safety-filtered dynamics preserve the prescribed pole region.

\begin{theorem}[Linearized pole-region preservation under CLM-CBF safety filtering]
\label{thm:linearized_pole_region_preservation}
Consider the safety-filtered closed-loop system \eqref{eq:clm_cbf_closed_loop_map} and its linearization \eqref{eq:linearized_safety_filtered_dynamics} around \((\mathbf{x}_e,\mathbf{r}^{\star}_e)\). Suppose that the nominal closed-loop matrix satisfies \eqref{eq:nominal_pole_region_certificate} for some \(P=P^T>0\), \(Q=Q^T>0\), and \(\sigma>0\). Assume further that the optimizer map \(\pi\) is locally differentiable in the operating region \(\Omega\) and satisfies \eqref{eq:optimizer_sensitivity_bound}. If
\begin{equation}
\left\|
B_{\mathrm{cl}}
\right\|
L_{\pi}
<
\frac{\lambda_{\min}(Q)}
{2\left\|P\right\|},
\label{eq:bounded_pole_region_condition}
\end{equation}
then all eigenvalues of the local active safety-filtered matrix \(A_{\mathrm{eff}}\) remain in the prescribed pole region \(\mathcal{D}_{\sigma}\).
\end{theorem}

\begin{proof}
The local active safety-filtered linearization is governed by
\begin{equation}
A_{\mathrm{eff}}
=
A_{\mathrm{cl}}
+
\Delta A_{\pi}.
\end{equation}
To guarantee that all eigenvalues of \(A_{\mathrm{eff}}\) remain in \(\mathcal{D}_{\sigma}\), it is sufficient to show that
\begin{equation}
A_{\mathrm{eff}}^T P
+
P A_{\mathrm{eff}}
+
2\sigma P
<
0.
\label{eq:perturbed_pole_region_certificate}
\end{equation}
Substituting \(A_{\mathrm{eff}}=A_{\mathrm{cl}}+\Delta A_{\pi}\) into \eqref{eq:perturbed_pole_region_certificate} gives
\begin{align}
&
A_{\mathrm{eff}}^T P
+
P A_{\mathrm{eff}}
+
2\sigma P
\nonumber\\
&=
A_{\mathrm{cl}}^T P
+
P A_{\mathrm{cl}}
+
2\sigma P
+
\Delta A_{\pi}^T P
+
P\Delta A_{\pi}.
\end{align}
Using the nominal certificate \eqref{eq:nominal_pole_region_certificate}, it follows that
\begin{align}
&
A_{\mathrm{eff}}^T P
+
P A_{\mathrm{eff}}
+
2\sigma P
\nonumber\\
&\leq
-Q
+
\Delta A_{\pi}^T P
+
P\Delta A_{\pi}.
\end{align}
Hence, \eqref{eq:perturbed_pole_region_certificate} is guaranteed if
\begin{equation}
\Delta A_{\pi}^T P
+
P\Delta A_{\pi}
<
Q.
\label{eq:perturbation_margin_condition}
\end{equation}
Furthermore,
\begin{equation}
\left\|
\Delta A_{\pi}^T P
+
P\Delta A_{\pi}
\right\|
\leq
2
\left\|P\right\|
\left\|
\Delta A_{\pi}
\right\|.
\end{equation}
From \eqref{eq:optimizer_induced_perturbation} and \eqref{eq:optimizer_sensitivity_bound},
\begin{equation}
\left\|
\Delta A_{\pi}
\right\|
\leq
\left\|
B_{\mathrm{cl}}
\right\|
L_{\pi}.
\end{equation}
Therefore, condition \eqref{eq:bounded_pole_region_condition} implies
\begin{equation}
2
\left\|P\right\|
\left\|
\Delta A_{\pi}
\right\|
<
\lambda_{\min}(Q),
\end{equation}
which is sufficient for \eqref{eq:perturbation_margin_condition}. Consequently, \eqref{eq:perturbed_pole_region_certificate} holds, and all eigenvalues of \(A_{\mathrm{eff}}\) lie in \(\mathcal{D}_{\sigma}\).
\end{proof}

\begin{remark}[Dependence on the prescribed decay rate]
Although the sufficient condition in \eqref{eq:bounded_pole_region_condition} does not contain \(\sigma\) explicitly, its influence enters through the shifted Lyapunov certificate \eqref{eq:nominal_pole_region_certificate}. For each prescribed decay rate \(\sigma\), the matrices \(P\) and \(Q\) must certify the shifted matrix \(A_{\mathrm{cl}}+\sigma I\). Hence, the admissible optimizer sensitivity is implicitly determined by the selected pole-region requirement.
\end{remark}

Theorem~\ref{thm:linearized_pole_region_preservation} provides a quantitative robustness condition for the linearized active CLM-CBF safety filter. It shows that the active safety filter does not need to preserve the nominal poles exactly. Instead, it is sufficient that the optimizer-induced first-order perturbation remains smaller than the available robustness reserve encoded by the Lyapunov certificate \eqref{eq:nominal_pole_region_certificate}. The corresponding admissible optimizer sensitivity bound is given by \eqref{eq:bounded_pole_region_condition}.

The preceding result certifies the pole region of the linearized active safety-filtered dynamics. To account for the nonlinear terms in \eqref{eq:exact_local_safety_filtered_dynamics}, we next impose a bound on the higher-order residual. Suppose that, in a neighborhood
\begin{equation}
\mathcal{B}_{\rho}
=
\Bigl\{
(\delta\mathbf{x},\delta\mathbf{r}^{\star})
:
\left\|
\begin{bmatrix}
\delta\mathbf{x}\\
\delta\mathbf{r}^{\star}
\end{bmatrix}
\right\|
\leq
\rho
\Bigr\},
\label{eq:local_validity_neighborhood}
\end{equation}
there exists a constant \(L_d>0\) such that
\begin{equation}
\left\|
\mathbf{d}
\!\left(
\delta\mathbf{x},
\delta\mathbf{r}^{\star}
\right)
\right\|
\leq
L_d
\left\|
\begin{bmatrix}
\delta\mathbf{x}\\
\delta\mathbf{r}^{\star}
\end{bmatrix}
\right\|^2
\label{eq:quadratic_residual_bound}
\end{equation}
for all \((\delta\mathbf{x},\delta\mathbf{r}^{\star})\in\mathcal{B}_{\rho}\). Such a bound follows, for example, if \(\mathbf{f}_{\mathrm{sf}}\) is continuously differentiable and its Jacobian is locally Lipschitz in the considered active-set region. The quadratic form of \eqref{eq:quadratic_residual_bound} reflects that the first-order Taylor terms have been explicitly removed in \eqref{eq:nonlinear_residual_definition}.

The following theorem shows that the linearized pole-region certificate also induces a local nonlinear validity region for the exact safety-filtered dynamics.

\begin{theorem}[Local nonlinear validity of the pole-region certificate]
\label{thm:local_nonlinear_validity}
Consider the exact local safety-filtered dynamics \eqref{eq:exact_local_safety_filtered_dynamics}. Suppose that all conditions of Theorem~\ref{thm:linearized_pole_region_preservation} hold, so that the linearized active dynamics satisfy
\begin{equation}
A_{\mathrm{eff}}^T P
+
P A_{\mathrm{eff}}
+
2\sigma P
\leq
-Q_{\mathrm{eff}},
\label{eq:Aeff_shifted_certificate}
\end{equation}
for some \(P=P^T>0\), \(\sigma>0\), and
\begin{equation}
Q_{\mathrm{eff}}
=
Q
-
2\|P\|\|B_{\mathrm{cl}}\|L_\pi I.
\label{eq:Qeff_definition}
\end{equation}
The strict sensitivity condition \eqref{eq:bounded_pole_region_condition} implies \(Q_{\mathrm{eff}}=Q_{\mathrm{eff}}^T>0\). Further suppose that the nonlinear residual satisfies
\begin{equation}
\left\|
\mathbf{d}
\!\left(
\delta\mathbf{x},
\mathbf{0}
\right)
\right\|
\leq
L_d
\left\|
\delta\mathbf{x}
\right\|^2
\label{eq:state_only_residual_bound}
\end{equation}
for all \(\delta\mathbf{x}\) satisfying
\begin{equation}
(\delta\mathbf{x},\mathbf{0})\in\mathcal{B}_{\rho},
\qquad
\text{i.e.}
\qquad
\left\|
\delta\mathbf{x}
\right\|
\leq
\rho .
\end{equation}
Then the equilibrium of the exact nonlinear safety-filtered dynamics is locally exponentially stable for constant reference commands. More precisely, consider the Lyapunov function
\begin{equation}
V(\delta\mathbf{x})
=
\frac{1}{2}
\delta\mathbf{x}^TP\delta\mathbf{x}.
\label{eq:local_nonlinear_lyapunov_function}
\end{equation}
For any radius \(\rho_x>0\) satisfying
\begin{equation}
\rho_x
<
\frac{
\frac{1}{2}\lambda_{\min}(Q_{\mathrm{eff}})
+
\sigma\lambda_{\min}(P)
}
{
\left\|P\right\|L_d
},
\label{eq:local_validity_radius}
\end{equation}
the derivative of \eqref{eq:local_nonlinear_lyapunov_function} along the exact nonlinear safety-filtered dynamics is strictly negative for all
\begin{equation}
0
<
\left\|
\delta\mathbf{x}
\right\|
\leq
\min\{\rho,\rho_x\}.
\end{equation}
\end{theorem}

\begin{proof}
For constant reference commands, \(\delta\mathbf{r}^{\star}=\mathbf{0}\), and the exact local dynamics reduce to
\begin{equation}
\delta\dot{\mathbf{x}}
=
A_{\mathrm{eff}}\delta\mathbf{x}
+
\mathbf{d}
\!\left(
\delta\mathbf{x},
\mathbf{0}
\right).
\end{equation}
Using the Lyapunov function \eqref{eq:local_nonlinear_lyapunov_function}, its derivative along the exact local dynamics satisfies
\begin{align}
\dot V
&=
\delta\mathbf{x}^TP
A_{\mathrm{eff}}\delta\mathbf{x}
+
\delta\mathbf{x}^TP
\mathbf{d}
\!\left(
\delta\mathbf{x},
\mathbf{0}
\right)
\nonumber\\
&=
\frac{1}{2}
\delta\mathbf{x}^T
\left(
A_{\mathrm{eff}}^TP
+
P A_{\mathrm{eff}}
\right)
\delta\mathbf{x}
+
\delta\mathbf{x}^TP
\mathbf{d}
\!\left(
\delta\mathbf{x},
\mathbf{0}
\right).
\end{align}
Using \eqref{eq:Aeff_shifted_certificate}, we obtain
\begin{equation}
A_{\mathrm{eff}}^TP
+
P A_{\mathrm{eff}}
\leq
-Q_{\mathrm{eff}}
-
2\sigma P.
\end{equation}
Hence,
\begin{align}
\dot V
&\leq
-
\frac{1}{2}
\delta\mathbf{x}^TQ_{\mathrm{eff}}\delta\mathbf{x}
-
\sigma
\delta\mathbf{x}^TP\delta\mathbf{x}
+
\left\|
\delta\mathbf{x}
\right\|
\left\|P\right\|
\left\|
\mathbf{d}
\!\left(
\delta\mathbf{x},
\mathbf{0}
\right)
\right\|.
\label{eq:Vdot_before_residual_bound}
\end{align}
For all \(\delta\mathbf{x}\) satisfying
\begin{equation}
(\delta\mathbf{x},\mathbf{0})\in\mathcal{B}_{\rho},
\end{equation}
the residual bound \eqref{eq:state_only_residual_bound} holds. Therefore,
\begin{align}
\dot V
&\leq
-
\left(
\frac{1}{2}\lambda_{\min}(Q_{\mathrm{eff}})
+
\sigma\lambda_{\min}(P)
\right)
\left\|
\delta\mathbf{x}
\right\|^2
+
\left\|P\right\|L_d
\left\|
\delta\mathbf{x}
\right\|^3
\nonumber\\
&=
-
\left[
\frac{1}{2}\lambda_{\min}(Q_{\mathrm{eff}})
+
\sigma\lambda_{\min}(P)
-
\left\|P\right\|L_d
\left\|
\delta\mathbf{x}
\right\|
\right]
\left\|
\delta\mathbf{x}
\right\|^2.
\label{eq:Vdot_final_bound}
\end{align}
If
\begin{equation}
0
<
\left\|
\delta\mathbf{x}
\right\|
\leq
\min\{\rho,\rho_x\},
\end{equation}
then the residual bound is valid and, by \eqref{eq:local_validity_radius}, the term in brackets in \eqref{eq:Vdot_final_bound} is strictly positive. Hence,
\begin{equation}
\dot V<0
\end{equation}
for all nonzero \(\delta\mathbf{x}\) in this neighborhood. This proves local exponential stability of the exact nonlinear safety-filtered dynamics for constant reference commands.
\end{proof}


\begin{remark}[Effect of reference variations]
Theorem~\ref{thm:local_nonlinear_validity} addresses the autonomous local behavior for constant references. If \(\delta\mathbf{r}^{\star}\neq \mathbf{0}\), then the exact dynamics \eqref{eq:exact_local_safety_filtered_dynamics} contain the additional input term \(B_{\mathrm{eff}}\delta\mathbf{r}^{\star}\). In that case, the same Lyapunov argument yields a local input-to-state stability type bound with respect to reference variations. In particular,
\begin{align}
\dot V
&\leq
-\alpha_x
\left\|
\delta\mathbf{x}
\right\|^2
+
\left\|P\right\|
\left\|
B_{\mathrm{eff}}
\right\|
\left\|
\delta\mathbf{x}
\right\|
\left\|
\delta\mathbf{r}^{\star}
\right\|
\nonumber\\
&\quad
+
\left\|P\right\|L_d
\left\|
\delta\mathbf{x}
\right\|
\left\|
\begin{bmatrix}
\delta\mathbf{x}\\
\delta\mathbf{r}^{\star}
\end{bmatrix}
\right\|^2,
\label{eq:local_iss_bound_reference}
\end{align}
where
\begin{equation}
\alpha_x
=
\frac{1}{2}\lambda_{\min}(Q_{\mathrm{eff}})
+
\sigma\lambda_{\min}(P).
\end{equation}
Thus, sufficiently small reference variations lead to bounded local deviations of the exact nonlinear safety-filtered dynamics. The pole-region certificate should therefore be interpreted as a local robustness certificate around the considered operating point, with the certified neighborhood determined by the Lyapunov decay margin and the nonlinear residual bound.
\end{remark}

\section{Pole-Region-Constrained CLM-CBF Safety Filter Synthesis}
\label{sec:Pole-Region-Constrained_CLM-CBF_Safety_Filter_Synthesis}

It remains to explain how the quantity \(L_{\pi}\) can be obtained and how the CLM-CBF QP can be designed such that \eqref{eq:bounded_pole_region_condition} is satisfied. We revisit the CLM-CBF QP formulation introduced in Section~\ref{sec:Closed_Loop_Model_CBFs}. The safety constraint in \eqref{eq:cl_reference_filter_qp} can be written in the compact form
\begin{equation}
a(\mathbf{x})\mathbf{r}
\geq
b(\mathbf{x}),
\label{eq:compact_affine_constraint}
\end{equation}
with
\begin{equation}
a(\mathbf{x})
=
\frac{\partial h}{\partial \mathbf{x}}B_{\mathrm{cl}},
\label{eq:a_definition}
\end{equation}
and
\begin{equation}
b(\mathbf{x})
=
-
\frac{\partial h}{\partial \mathbf{x}}
A_{\mathrm{cl}}\mathbf{x}
-
\alpha(h(\mathbf{x})).
\label{eq:b_definition}
\end{equation}
The optimizer map is therefore determined by a parametric projection of the desired reference \(\mathbf{r}^{\star}\) onto the admissible reference set defined by \eqref{eq:compact_affine_constraint}.

The CLM-CBF filter is inactive whenever
\begin{equation}
a(\mathbf{x})\mathbf{r}^{\star}
\geq
b(\mathbf{x}).
\label{eq:inactive_reference_condition}
\end{equation}
In this case, \(\mathbf{r}^{\star}\) is feasible and minimizes the objective of \eqref{eq:cl_reference_filter_qp}. Hence,
\begin{equation}
\pi(\mathbf{x},\mathbf{r}^{\star})
=
\mathbf{r}^{\star},
\end{equation}
and therefore
\begin{equation}
\frac{\partial \pi}{\partial \mathbf{x}}
=
0.
\label{eq:inactive_optimizer_sensitivity}
\end{equation}
Consequently, the nominal closed-loop dynamics are recovered exactly whenever the safety filter is inactive. The only relevant case for pole-region analysis is therefore the active regime, in which
\begin{equation}
a(\mathbf{x})\mathbf{r}^{\star}
<
b(\mathbf{x}).
\label{eq:active_reference_condition}
\end{equation}

For the single-constraint case, the active-set solution of \eqref{eq:cl_reference_filter_qp} is available in closed form. Defining
\begin{equation}
\beta(\mathbf{x})
=
a(\mathbf{x})W^{-1}a(\mathbf{x})^T,
\label{eq:beta_definition}
\end{equation}
and
\begin{equation}
\eta(\mathbf{x},\mathbf{r}^{\star})
=
b(\mathbf{x})
-
a(\mathbf{x})\mathbf{r}^{\star},
\label{eq:eta_definition}
\end{equation}
the optimizer in the active region is
\begin{equation}
\pi(\mathbf{x},\mathbf{r}^{\star})
=
\mathbf{r}^{\star}
+
\frac{\eta(\mathbf{x},\mathbf{r}^{\star})}
{\beta(\mathbf{x})}
W^{-1}a(\mathbf{x})^T.
\label{eq:active_optimizer_solution}
\end{equation}
Equation~\eqref{eq:active_optimizer_solution} shows explicitly how the CLM-CBF correction depends on the state. This dependence is the source of both the first-order optimizer-induced perturbation \(\Delta A_{\pi}\) and, if \(a(\mathbf{x})\), \(b(\mathbf{x})\), or the active set vary nonlinearly over the operating region, the higher-order residual \(\mathbf{d}(\delta\mathbf{x},\delta\mathbf{r}^{\star})\).

The local optimizer sensitivity can be obtained by differentiating \eqref{eq:active_optimizer_solution}. Let
\begin{equation}
q(\mathbf{x},\mathbf{r}^{\star})
=
\frac{\eta(\mathbf{x},\mathbf{r}^{\star})}
{\beta(\mathbf{x})}.
\end{equation}
Then
\begin{equation}
\pi(\mathbf{x},\mathbf{r}^{\star})
=
\mathbf{r}^{\star}
+
q(\mathbf{x},\mathbf{r}^{\star})
W^{-1}a(\mathbf{x})^T.
\end{equation}
Assuming that \(\mathbf{r}^{\star}\) is exogenous with respect to \(\mathbf{x}\), differentiation with respect to \(\mathbf{x}\) gives
\begin{equation}
\frac{\partial \pi}{\partial \mathbf{x}}
=
W^{-1}a(\mathbf{x})^T
\frac{\partial q}{\partial \mathbf{x}}
+
q(\mathbf{x},\mathbf{r}^{\star})
W^{-1}
\frac{\partial a(\mathbf{x})^T}{\partial \mathbf{x}}.
\label{eq:general_optimizer_sensitivity}
\end{equation}
Furthermore,
\begin{equation}
\frac{\partial q}{\partial \mathbf{x}}
=
\frac{
\beta(\mathbf{x})
\frac{\partial \eta}{\partial \mathbf{x}}
-
\eta(\mathbf{x},\mathbf{r}^{\star})
\frac{\partial \beta}{\partial \mathbf{x}}
}
{\beta(\mathbf{x})^2},
\label{eq:q_sensitivity}
\end{equation}
where
\begin{equation}
\frac{\partial \eta}{\partial \mathbf{x}}
=
\frac{\partial b}{\partial \mathbf{x}}
-
\mathbf{r}^{\star T}
\frac{\partial a^T}{\partial \mathbf{x}},
\label{eq:eta_sensitivity}
\end{equation}
and
\begin{equation}
\frac{\partial \beta}{\partial \mathbf{x}}
=
2a(\mathbf{x})W^{-1}
\frac{\partial a(\mathbf{x})^T}{\partial \mathbf{x}}.
\label{eq:beta_sensitivity}
\end{equation}
Equations~\eqref{eq:general_optimizer_sensitivity}--\eqref{eq:beta_sensitivity} provide an explicit analytical expression for the sensitivity of the CLM-CBF optimizer map in the active region.

A particularly important special case arises when the safety constraint is affine in the state. We consider
\begin{equation}
h(\mathbf{x})
=
h_0
-
C_h\mathbf{x},
\label{eq:affine_barrier_function}
\end{equation}
where \(C_h\in\mathbb{R}^{1\times n}\). Then
\begin{equation}
\frac{\partial h}{\partial \mathbf{x}}
=
-C_h
\end{equation}
is constant, and therefore
\begin{equation}
a
=
-C_hB_{\mathrm{cl}}
\end{equation}
is also constant. Consequently,
\begin{equation}
\frac{\partial a}{\partial \mathbf{x}}
=
0,
\qquad
\frac{\partial \beta}{\partial \mathbf{x}}
=
0.
\end{equation}
If the extended class-\(\mathcal{K}\) function is chosen as
\begin{equation}
\alpha(h)=\gamma h,
\qquad
\gamma>0,
\label{eq:linear_alpha}
\end{equation}
then
\begin{equation}
b(\mathbf{x})
=
C_hA_{\mathrm{cl}}\mathbf{x}
-
\gamma
\left(
h_0-C_h\mathbf{x}
\right),
\end{equation}
and hence
\begin{equation}
\frac{\partial b}{\partial \mathbf{x}}
=
C_hA_{\mathrm{cl}}
+
\gamma C_h.
\label{eq:b_sensitivity_affine}
\end{equation}
The optimizer sensitivity in the active region reduces to
\begin{equation}
\frac{\partial \pi}{\partial \mathbf{x}}
=
\frac{
W^{-1}a^T
}{
aW^{-1}a^T
}
\left(
C_hA_{\mathrm{cl}}
+
\gamma C_h
\right).
\label{eq:affine_constraint_optimizer_sensitivity}
\end{equation}
Thus, for the affine single-constraint case, an analytical sensitivity bound is
\begin{equation}
L_{\pi}
=
\left\|
\frac{
W^{-1}a^T
}{
aW^{-1}a^T
}
\left(
C_hA_{\mathrm{cl}}
+
\gamma C_h
\right)
\right\|.
\label{eq:affine_constraint_Lpi}
\end{equation}
Substitution of \eqref{eq:affine_constraint_Lpi} into \eqref{eq:bounded_pole_region_condition} gives an explicit design condition for the CLM-CBF QP parameters:
\begin{equation}
\left\|
\frac{
W^{-1}a^T
}{
aW^{-1}a^T
}
\left(
C_hA_{\mathrm{cl}}
+
\gamma C_h
\right)
\right\|
\leq
\frac{\lambda_{\min}(Q)}
{2\left\|P\right\|
\left\|B_{\mathrm{cl}}\right\|}.
\label{eq:explicit_qp_design_condition}
\end{equation}

Several observations follow from \eqref{eq:explicit_qp_design_condition}. First, the influence of the QP weighting matrix \(W\) enters through the weighted projection direction
\begin{equation}
\frac{
W^{-1}a^T
}{
aW^{-1}a^T
}.
\end{equation}
Therefore, the relative weighting between reference channels affects how the CLM-CBF correction is distributed among the available reference directions. In contrast, a uniform scalar scaling \(W\mapsto \rho W\), \(\rho>0\), cancels in this expression for the single-constraint case and therefore does not change the optimizer sensitivity. Hence, it is the geometry of \(W\), not merely its magnitude, that is relevant for pole-region preservation.

Second, the CBF gain \(\gamma\) appears directly in the factor
\begin{equation}
C_hA_{\mathrm{cl}}
+
\gamma C_h.
\end{equation}
Thus, \(\gamma\) affects both the safety-filter activation behavior and the sensitivity of the optimizer map. A large value of \(\gamma\) may reduce conservatism away from the boundary, but may also increase the local sensitivity once the constraint becomes active. Conversely, a smaller value of \(\gamma\) may reduce the optimizer sensitivity, but can lead to earlier or more restrictive safety intervention. Therefore, \(\gamma\) should not be selected solely from a safety perspective, but also with respect to the admissible sensitivity bound \eqref{eq:bounded_pole_region_condition}.

Third, the condition \eqref{eq:explicit_qp_design_condition} connects the safety-filter parameters directly to the robustness properties of the nominal closed-loop system. The matrices \(P\) and \(Q\) are obtained from the pole-region certificate of the robust controller, while \(L_{\pi}\) is induced by the CLM-CBF QP. Hence, the safety-filter design becomes a robustness-constrained synthesis problem.

Similarly to Section~\ref{sec:CL_Controller_Design}, we obtain a constrained multi-objective optimization problem. With respect to the safety filter, the following offline synthesis problem is considered:
\begin{equation}
\begin{aligned}
\min_{W,\gamma}
\quad &
J_{\mathrm{sf}}(W,\gamma)
\\
\mathrm{s.t.}
\quad &
W=W^T>0,
\\
&
\gamma>0,
\\
&
L_{\pi}(W,\gamma)
\leq
L_{\pi,\max},
\\
&
L_d(W,\gamma)
\leq
L_{d,\max},
\\
&
\mathcal{F}_{\mathrm{safe}}(W,\gamma)\neq\emptyset.
\end{aligned}
\label{eq:offline_safety_filter_synthesis_static}
\end{equation}
Here, \(J_{\mathrm{sf}}\) denotes a design objective for the safety filter, for example penalizing reference deviation, excessive intervention, or loss of command fidelity. The condition
\begin{equation}
\mathcal{F}_{\mathrm{safe}}(W,\gamma)\neq\emptyset
\end{equation}
denotes feasibility of the CLM-CBF QP over the considered operating region. The constraint
\begin{equation}
L_{\pi}(W,\gamma)
\leq
L_{\pi,\max}
\end{equation}
guarantees, through Theorem~\ref{thm:linearized_pole_region_preservation}, that the first-order optimizer-induced perturbation remains inside the available pole-region robustness margin. The additional constraint on \(L_d(W,\gamma)\), with the upper bound $L_{d,\max}$ to be chosen, accounts for the higher-order nonlinear residual and ensures that the certified local validity radius in \eqref{eq:local_validity_radius} remains sufficiently large for the operating region of interest. Thus, the safety-filter synthesis problem jointly balances safety enforcement, reference fidelity, pole-region preservation, and local nonlinear robustness.
\begin{remark}[Computation of optimizer sensitivity and nonlinear residual bounds]
For more general barrier functions, or for cases in which multiple independent safety constraints are considered simultaneously, a closed-form expression such as \eqref{eq:affine_constraint_Lpi} is generally not available. In that case, the optimizer sensitivity \(L_{\pi}\) can be computed or approximated through parametric QP sensitivity analysis. Moreover, the same local sensitivity information can be used to estimate the higher-order residual bound \(L_d\) appearing in \eqref{eq:quadratic_residual_bound}.

Consider the parametric QP
\begin{equation}
\begin{aligned}
\pi(\mathbf{x},\mathbf{r}^{\star})
=
\arg\min_{\mathbf{r}}
\quad &
\frac{1}{2}
\left(
\mathbf{r}-\mathbf{r}^{\star}
\right)^T
W
\left(
\mathbf{r}-\mathbf{r}^{\star}
\right)
\\
\mathrm{s.t.}
\quad &
A_c(\mathbf{x})\mathbf{r}
\geq
b_c(\mathbf{x}),
\end{aligned}
\label{eq:general_parametric_qp}
\end{equation}
where \(A_c(\mathbf{x})\) collects the constraint directions and \(b_c(\mathbf{x})\) the corresponding lower bounds. Let \(\mathcal{A}\) denote a locally fixed active set. Under the usual regularity conditions for parametric quadratic programs, including uniqueness of the optimizer and nonsingularity of the active-set KKT matrix, the optimizer map is locally differentiable within regions in which the active set remains unchanged.

For a fixed active set \(\mathcal{A}\), the active constraints can be treated as equalities. The corresponding KKT conditions are
\begin{equation}
W(\mathbf{r}-\mathbf{r}^{\star})
-
A_{\mathcal{A}}(\mathbf{x})^T\boldsymbol{\lambda}
=
0,
\label{eq:kkt_stationarity}
\end{equation}
\begin{equation}
A_{\mathcal{A}}(\mathbf{x})\mathbf{r}
-
b_{\mathcal{A}}(\mathbf{x})
=
0.
\label{eq:kkt_active_constraints}
\end{equation}

Here, \(A_{\mathcal A}(\mathbf{x})\) denotes the matrix obtained by selecting from \(A_c(\mathbf{x})\) only those rows whose indices belong to the active set \(\mathcal A\). Similarly, \(b_{\mathcal A}(\mathbf{x})\) denotes the vector obtained by selecting from \(b_c(\mathbf{x})\) only those components whose indices belong to \(\mathcal A\). 

Differentiating \eqref{eq:kkt_stationarity}--\eqref{eq:kkt_active_constraints} in a direction \(\delta\mathbf{x}\), while keeping \(\mathbf{r}^{\star}\) fixed, gives
\begin{equation}
\begin{bmatrix}
W & -A_{\mathcal{A}}^T\\
A_{\mathcal{A}} & 0
\end{bmatrix}
\begin{bmatrix}
\delta\mathbf{r}\\
\delta\boldsymbol{\lambda}
\end{bmatrix}
=
\begin{bmatrix}
\left(D_{\mathbf{x}}A_{\mathcal{A}}^T[\delta\mathbf{x}]\right)\boldsymbol{\lambda}\\
D_{\mathbf{x}}b_{\mathcal{A}}[\delta\mathbf{x}]
-
\left(D_{\mathbf{x}}A_{\mathcal{A}}[\delta\mathbf{x}]\right)\mathbf{r}
\end{bmatrix}.
\label{eq:kkt_sensitivity_system}
\end{equation}

The mappings \(D_{\mathbf{x}}A_{\mathcal A}[\delta\mathbf{x}]\) and \(D_{\mathbf{x}}b_{\mathcal A}[\delta\mathbf{x}]\) denote the directional derivatives of \(A_{\mathcal A}(\mathbf{x})\) and \(b_{\mathcal A}(\mathbf{x})\), respectively, with respect to \(\mathbf{x}\) in the direction \(\delta\mathbf{x}\).

Solving \eqref{eq:kkt_sensitivity_system} for a basis of perturbation directions yields the Jacobian
\begin{equation}
\frac{\partial \pi}{\partial \mathbf{x}}.
\end{equation}
The sensitivity bound can then be evaluated as
\begin{equation}
L_{\pi}
=
\sup_{\mathbf{x}\in\Omega}
\left\|
\frac{\partial \pi}{\partial \mathbf{x}}
\right\|.
\label{eq:Lpi_numerical_sampling}
\end{equation}

If the nonlinear residual bound in \eqref{eq:quadratic_residual_bound} is also required, it can be estimated from the deviation between the exact optimizer map and its first-order approximation. For each sampled point \((\mathbf{x}_e,\mathbf{r}^{\star}_e)\) and perturbation \((\delta\mathbf{x},\delta\mathbf{r}^{\star})\), define
\begin{align}
\mathbf{d}
\!\left(
\delta\mathbf{x},
\delta\mathbf{r}^{\star}
\right)
&=
\mathbf{f}_{\mathrm{sf}}
\!\left(
\mathbf{x}_e+\delta\mathbf{x},
\mathbf{r}^{\star}_e+\delta\mathbf{r}^{\star}
\right)
\nonumber\\
&\quad
-
A_{\mathrm{eff}}\delta\mathbf{x}
-
B_{\mathrm{eff}}\delta\mathbf{r}^{\star}.
\end{align}
Then a numerical local estimate of \(L_d\) can be obtained from
\begin{equation}
L_d
\approx
\sup_
{
(\delta\mathbf{x},\delta\mathbf{r}^{\star})\in\mathcal{B}_{\rho}
}
\frac{
\left\|
\mathbf{d}
\!\left(
\delta\mathbf{x},
\delta\mathbf{r}^{\star}
\right)
\right\|
}{
\left\|
\begin{bmatrix}
\delta\mathbf{x}\\
\delta\mathbf{r}^{\star}
\end{bmatrix}
\right\|^2
}.
\label{eq:Ld_numerical_estimate}
\end{equation}
This estimate should be evaluated only over neighborhoods in which the active set remains unchanged, or otherwise augmented by a conservative safety factor to account for active-set transitions. In practice, one may sample the operating region \(\Omega\), solve the QP, identify the active set, compute the first-order sensitivity through \eqref{eq:kkt_sensitivity_system}, and then evaluate both \(L_{\pi}\) and \(L_d\) over local perturbation neighborhoods. The resulting bounds are then used in \eqref{eq:offline_safety_filter_synthesis_static} to certify both linearized pole-region preservation and local nonlinear validity of the safety-filtered dynamics.
\end{remark}

%% file: Sections/Application.tex
The closed-loop model-based safety filter introduced in the previous subsections is now specialized to the flight envelope protection problem of the considered longitudinal missile dynamics. As established in Section~\ref{sec:Problem_Formulation_Missile}, the admissible operating region is defined by the angle-of-attack constraints in \eqref{eq:alpha_constraint}. In accordance with \cite{Autenrieb_2025}, the upper and lower bounds on $\alpha(t)$ are enforced via two distinct CBFs. While the simultaneous presence of multiple constraints may in general lead to feasibility issues in the associated quadratic program, the structure of the FEP problem mitigates this effect. In particular, the system typically evolves toward either the upper or lower boundary depending on the direction of motion, such that only one constraint is active at a given time.

To enforce these constraints defined in \eqref{eq:alpha_constraint}, the following CBF candidates are introduced,
\begin{equation}
h_{\alpha,u}(\mathbf{x}(t),\sigma_0)=\alpha_{\max}(\sigma_0)-\alpha(t),
\end{equation}
\begin{equation}
h_{\alpha,l}(\mathbf{x}(t),\sigma_0)=\alpha(t)-\alpha_{\min}(\sigma_0),
\end{equation}
ensuring that the system stays within the safe set
\begin{equation}
S_\alpha(\sigma_0)
=
\Bigl\{
\mathbf{x}(t) \in \mathbb{R}^2 \;\Bigm|\;
\alpha_{\min}(\sigma_0) \leq \alpha(t) \leq \alpha_{\max}(\sigma_0)
\Bigr\}.
\end{equation}
In particular, the safety condition is satisfied if \(h_{\alpha,u}(\mathbf{x}(t),\sigma_0) \geq 0\) and \( h_{\alpha,l}(\mathbf{x}(t),\sigma_0) \geq 0\), which implies that the system state remains in the intersection of the sets defined by the two CBF candidates for all \(t \geq 0\).

By using the concepts of Section~\ref{sec:Closed_Loop_Model_CBFs}, the CLM-CBF-based approach enforces the flight envelope constraints directly at the reference-level. Since the constraints are scalar and affine, the optimization problem remains convex and computationally efficient. At the same time, the underlying controller is left unchanged, such that its stability and performance properties are fully preserved. Applying the closed-loop CBF condition from Theorem~\ref{thm:robust_cl_cbf} to the system \eqref{eq:closed_loop_system_filter} yields two scalar CLM-CBF constraints of the form
\begin{equation*}
\begin{aligned}
\frac{\partial h_{\alpha,u}(\mathbf{x}(t),\sigma_0)}{\partial \mathbf{x}}
&\left(A_{\mathrm{cl}}(\sigma_0) \mathbf{x}(t) + B_{\mathrm{cl}}(\sigma_0) \mathbf{r}(t) \right)\\
&\geq -\gamma_{\alpha,u}\,h_{\alpha,u}(\mathbf{x}(t),\sigma_0),
\end{aligned}
\end{equation*}
\begin{equation*}
\begin{aligned}
\frac{\partial h_{\alpha,l}(\mathbf{x}(t),\sigma_0)}{\partial \mathbf{x}}
&\left(A_{\mathrm{cl}}(\sigma_0) \mathbf{x}(t)  + B_{\mathrm{cl}}(\sigma_0) \mathbf{r}(t) \right)\\
&\geq -\gamma_{\alpha,l}\,h_{\alpha,l}(\mathbf{x}(t),\sigma_0),
\end{aligned}
\end{equation*}
with $\gamma_{\alpha,u},\gamma_{\alpha,l}>0$.

In addition to the angle-of-attack constraints, it is also common and practically relevant in flight control applications to impose bounds on the pitch rate \(q(t)\). While the angle of attack directly characterizes the aerodynamic operating regime, the pitch rate reflects the rotational dynamics of the system and is closely related to actuator limitations and achievable control performance. In particular, excessive pitch rates may lead to actuator saturation or induce undesirable transient behavior. Therefore, it is beneficial to explicitly constrain \(q(t)\) within prescribed limits.

Analogous to the angle-of-attack constraints, the following CBF candidates are introduced:
\begin{equation}
h_{q,u}(\mathbf{x}(t),\sigma_0) = q_{\max}(\sigma_0) - q(t),
\end{equation}
\begin{equation}
h_{q,l}(\mathbf{x}(t),\sigma_0) = q(t) - q_{\min}(\sigma_0),
\end{equation}
defining the corresponding admissible set
\begin{equation}
S_q
\triangleq 
\Bigl\{
\mathbf{x}(t) \in \mathbb{R}^2 \;\Bigm|\;
q_{\min} \leq q(t) \leq q_{\max}
\Bigr\}.
\end{equation}

Applying the CLM-CBF concept to the pitch rate constraints yields
\begin{equation*}
\begin{aligned}
\frac{\partial h_{q,u}(\mathbf{x}(t),\sigma_0)}{\partial \mathbf{x}}
&\left( A_{\mathrm{cl}}(\sigma_0) \mathbf{x}(t) + B_{\mathrm{cl}}(\sigma_0) \mathbf{r}(t) \right)\\
&\geq -\gamma_{q,u}\, h_{q,u}(\mathbf{x}(t),\sigma_0),
\end{aligned}
\end{equation*}
\begin{equation*}
\begin{aligned}
\frac{\partial h_{q,l}(\mathbf{x}(t),\sigma_0)}{\partial \mathbf{x}}
&\left( A_{\mathrm{cl}}(\sigma_0) \mathbf{x}(t) + B_{\mathrm{cl}}(\sigma_0) \mathbf{r}(t) \right)\\
&\geq -\gamma_{q,l}\, h_{q,l}(\mathbf{x}(t),\sigma_0),
\end{aligned}
\end{equation*}
with \(\gamma_{q,u}, \gamma_{q,l} > 0\).

Together with the previously introduced angle-of-attack constraints, this results in a multi-constraint safety formulation, where forward invariance must be ensured with respect to the intersection of all individual safe sets. Consequently, the admissible operating region is given by the intersection \(S_\alpha(\sigma_0) \cap S_q\).

In addition to the state constraints, it is necessary to ensure that the certified reference signal \(r(t)\) remains compatible with the physical actuator limitations. Although the CLM-CBF formulation enforces safety with respect to the closed-loop state dynamics, the resulting control input is still generated through the feedback law in \eqref{eqn:robust_controll_law}, such that the choice of \(r(t)\) directly determines the actuator command. Consequently, the reference signal cannot be selected arbitrarily, but must be restricted to values that yield admissible control inputs.

Let the actuator be subject to magnitude constraints of the form
\begin{equation*}
\mathbf{u}_{\min}(\sigma_0) \leq \mathbf{u}(t) \leq \mathbf{u}_{\max}(\sigma_0).
\end{equation*}
Substituting the control law \eqref{eqn:robust_controll_law} yields the affine constraint
\begin{equation*}
\mathbf{u}_{\min}(\sigma_0) \leq K_x(\sigma_0)\mathbf{x}(t) + K_r(\sigma_0)\mathbf{r}(t) \leq \mathbf{u}_{\max}(\sigma_0).
\end{equation*}

This relation defines a state-dependent admissible set in the reference variable,
\begin{equation*}
\begin{aligned}
\mathcal{R}_u(\mathbf{x}(t),\sigma_0)
\triangleq
\Bigl\{
\mathbf{r}(t) \in \mathbb{R}^p \;\Bigm|\;
&\mathbf{u}_{\min}(\sigma_0) \leq  K_x(\sigma_0)\mathbf{x}(t) +\\ 
& K_r(\sigma_0)\mathbf{r}(t) \leq \mathbf{u}_{\max}(\sigma_0)
\Bigr\},
\end{aligned}
\end{equation*}
which characterizes all references that can be realized by the controller-actuator combination without violating input constraints.

Since the considered reference signal is scalar, i.e., \(\mathbf{r}(t)\in\mathbb{R}\), the actuator-induced admissible reference set can be expressed directly as an interval constraint. In particular, the inequality
\begin{equation*}
\mathbf{u}_{\min}(\sigma_0) \leq K_x(\sigma_0)\mathbf{x}(t) + K_r(\sigma_0)\mathbf{r}(t) \leq \mathbf{u}_{\max}(\sigma_0)
\end{equation*}
can be solved explicitly for \(\mathbf{r}(t)\). If \(K_r(\sigma_0)>0\), division by \(K_r(\sigma_0)\) preserves the inequality direction, such that
\begin{equation*}
\frac{\mathbf{u}_{\min}(\sigma_0)-K_x(\sigma_0)x(t)}{K_r(\sigma_0)}
\leq
\mathbf{r}(t)
\leq
\frac{\mathbf{u}_{\max}(\sigma_0)-K_x(\sigma_0)x(t)}{K_r(\sigma_0)}.
\end{equation*}
By contrast, if \(K_r(\sigma_0)<0\), division by \(K_r(\sigma_0)\) reverses the inequality direction, yielding
\begin{equation*}
\frac{\mathbf{u}_{\max}(\sigma_0)-K_x(\sigma_0)x(t)}{K_r(\sigma_0)}
\leq
\mathbf{r}(t)
\leq
\frac{\mathbf{u}_{\min}(\sigma_0)-K_x(\sigma_0)x(t)}{K_r(\sigma_0)}.
\end{equation*}

Hence, in the considered scalar-reference setting, the actuator magnitude limits induce explicit lower and upper bounds on the admissible reference signal, which can be incorporated directly into the quadratic program. We define the magnitude constraints on the input as:
\begin{equation*}
\mathbf{r}_{\max}(\mathbf{x}(t),\sigma_0)
=
\frac{\mathbf{u}_{\max}(\sigma_0)-K_x(\sigma_0)\mathbf{x}(t)}{K_r(\sigma_0)},
\end{equation*}
\begin{equation*}
\mathbf{r}_{\min}(\mathbf{x}(t),\sigma_0)
=
\frac{\mathbf{u}_{\min}(\sigma_0)-K_x(\sigma_0)\mathbf{x}(t)}{K_r(\sigma_0)}.
\end{equation*}

The actuator magnitude constraints can be written in the compact form
\begin{equation*}
\mathbf{r}_{\min}(\mathbf{x}(t),\sigma_0)\leq \mathbf{r}(t)\leq \mathbf{r}_{\max}(\mathbf{x}(t),\sigma_0).
\end{equation*}
Thus, in the considered scalar-reference setting, the actuator limits induce explicit upper and lower bounds on the admissible reference signal, which can be incorporated directly into the quadratic program as affine inequality constraints.

Based on the previously derived CLM-CBF conditions, the reference-level safety filter is formulated as the following quadratic program:
\begin{align}
\mathbf{r}_{\mathrm{safe}}(t)
= \arg\min_{\mathbf{r}} &\quad  \frac{1}{2}
\left(
\mathbf{r}-\mathbf{r}^{\star}
\right)^T
W
\left(
\mathbf{r}-\mathbf{r}^{\star}
\right) \label{eq:CLM_QP_results}\\
\text{s.t.} \quad
& \dot{h}_{\alpha,u}
\geq -\gamma_{\alpha,u}\,h_{\alpha,u}(\mathbf{x}(t),\sigma_0), \notag\\
& \dot{h}_{\alpha,l}
\geq -\gamma_{\alpha,l}\,h_{\alpha,l}(\mathbf{x}(t),\sigma_0), \notag\\
& \dot{h}_{q,u}
\geq -\gamma_{q,u}\,h_{q,u}(\mathbf{x}(t)), \notag\\
& \dot{h}_{q,l}
\geq -\gamma_{q,l}\,h_{q,l}(\mathbf{x}(t)), \notag\\
& \mathbf{r}_{\min}(\mathbf{x}(t),\sigma_0)\leq \mathbf{r}(t)\leq \mathbf{r}_{\max}(\mathbf{x}(t),\sigma_0), \notag
\end{align}
where
\begin{align*}
\dot{h}_{\alpha,u}
&:= \frac{\partial h_{\alpha,u}(\mathbf{x}(t),\sigma_0)}{\partial \mathbf{x}}
\left(A_{\mathrm{cl}}(\sigma_0)\mathbf{x}(t)+B_{\mathrm{cl}}(\sigma_0)\mathbf{r}(t)\right),\\[0.5em]
\dot{h}_{\alpha,l}
&:= \frac{\partial h_{\alpha,l}(\mathbf{x}(t),\sigma_0)}{\partial \mathbf{x}}
\left(A_{\mathrm{cl}}(\sigma_0)\mathbf{x}(t)+B_{\mathrm{cl}}(\sigma_0)\mathbf{r}(t)\right),\\[0.5em]
\dot{h}_{q,u}
&:= \frac{\partial h_{q,u}(\mathbf{x}(t),\sigma_0)}{\partial \mathbf{x}}
\left(A_{\mathrm{cl}}(\sigma_0)\mathbf{x}(t)+B_{\mathrm{cl}}(\sigma_0)\mathbf{r}(t)\right),\\[0.5em]
\dot{h}_{q,l}
&:= \frac{\partial h_{q,l}(\mathbf{x}(t),\sigma_0)}{\partial \mathbf{x}}
\left(A_{\mathrm{cl}}(\sigma_0)\mathbf{x}(t)+B_{\mathrm{cl}}(\sigma_0)\mathbf{r}(t)\right).
\end{align*}

The choice of the CLM-CBF design parameters $W$, $\gamma_{\alpha,u}$,$\gamma_{\alpha,l}$,$\gamma_{q,u}$ and $\gamma_{q,l}$ is based on the proposed safety filter synthesis from Section~\ref{sec:Pole-Region-Constrained_CLM-CBF_Safety_Filter_Synthesis}.

\subsection{Numerical Results}
\label{sec:Numerical Results}

In the following, the proposed CLM-CBF framework from Section~\ref{sec:Closed_Loop_Model_CBFs} is evaluated for the flight envelope protection problem introduced in Section~\ref{sec:Problem_Formulation_Missile}. The considered plant is the longitudinal missile model discussed in Section~\ref{Sect: Nonlinear_Missile_Model}, where the controlled states are the angle of attack $\alpha(t)$ and the pitch rate $q(t)$, and the fin deflection $\delta(t)$ acts as control input. The nominal tracking controller used throughout the simulations is the state-feedback controller introduced in Section~\ref{sec:Closed_Loop_Model_CBFs}. The purpose of the following case study is to demonstrate that the proposed CLM-CBF formulation provides the same safety behavior as the classical input-level formulation, without weakening robustness margins or altering the desired closed-loop dynamics. 

In the first step, we assess the CLM-CBF-based safety filtering, which was designed as proposed in Section~\ref{sec:Pole-Region-Constrained_CLM-CBF_Safety_Filter_Synthesis}, with respect to the closed-loop robustness properties. To do so, we revisit the analysis previously conducted in Section~\ref{sec:Motivating Example} for a simplified benchmark system and now apply it to the full missile model, controlled by the nominal feedback law \eqref{eqn:robust_controll_law} designed according to the robust control design outlined in Section~\ref{sec:Closed_Loop_Model_CBFs}. The resulting controller ensures satisfactory tracking performance together with meaningful robustness margins in the nominal case. In particular, the margins of the proposed gains at the considered operating point $\sigma_0$ correspond to a gain margin of $\pm 16.4\;\mathrm{dB}$ and a phase margin of $\pm 72.7877^\circ$, confirming that the nominal closed-loop system possesses a reasonable degree of robustness.

Similar to the previous analysis, a reduced FEP problem is considered in which only pitch-rate constraints are enforced. This allows for a simpler analysis and a direct comparison with the simplified case analyzed previously. The pitch-rate constraint is defined as
\begin{equation*}
    q_{\min}(\sigma_0) \le q(t) \le q_{\max}(\sigma_0),
\end{equation*}
with $q_{\min}(\sigma_0)=-30^\circ/\mathrm{s}$ and $q_{\max}(\sigma_0)=30^\circ/\mathrm{s}$. The corresponding control barrier function candidates are introduced as
\begin{equation*}
    h_{q,u}(\mathbf{x}(t))=q_{\max}(\sigma_0)-q, \qquad h_{q,l}(\mathbf{x}(t))=q-q_{\min}(\sigma_0).
\end{equation*}
Since the constraint has relative degree one, the input-level CBF safety filter is formulated as the quadratic program
\begin{align*}
    \mathbf{u}_{\mathrm{safe}}(t)=\arg\min_u \quad & \|\mathbf{u}(t)-\mathbf{u}^\star(t)\|_2^2 \notag\\
    \text{s.t.}\quad
    & \dot h_{q,u}(\mathbf{x}(t),\mathbf{u}(t))+\gamma_{q,u} h_{q,u}(\mathbf{x}(t)) \ge 0, \notag\\
    & \dot h_{q,l}(\mathbf{x}(t),\mathbf{u}(t))+\gamma_{q,l} h_{q,l}(\mathbf{x}(t)) \ge 0,
    \label{eq:qp_q_only_results}
\end{align*}
where $\mathbf{u}^\star(t)$ denotes the nominal control input. This formulation corresponds directly to the one used in the motivating example of Section~\ref{sec:Motivating Example}.

Two representative operating conditions are considered. First, an interior point of the safe set is selected, corresponding to a configuration in which the safety filter remains inactive. This point is identified through a scan of the admissible set and corresponds to $\mathbf{x}_{\mathrm{inactive}} = [8.1^\circ, -2.4^\circ/\mathrm{s}]^\top$, for which both barrier constraints are strictly satisfied and the associated Lagrange multipliers are zero. Linearization of the filtered closed-loop dynamics around this point and the analysis of the margins shows that the system has pole locations corresponding to a gain margin of $\pm 16.4\;\mathrm{dB}$ and a phase margin of $\pm 72.7877^\circ$, which is the same as for the nominal system without any safety filter. These results confirm that, as long as the system operates sufficiently far from the set boundaries, the safety filter remains inactive and the closed-loop behavior is indistinguishable from the closed-loop system without a CLM-CBF-based safety filter. Next, an active constraint scenario is considered by initializing the system directly on the upper pitch-rate boundary, i.e., $\mathbf{x}_{\mathrm{active}}=[-12^\circ, 30^\circ/\mathrm{s}]^\top$. In this condition, the relevant pitch rate constraint becomes active, leading to a non-zero Lagrange multiplier. The resulting linearized closed-loop poles indicate a substantial deviation from the nominal closed-loop dynamics. In particular, the oscillatory pole pair observed in the nominal case collapses into two real poles, reflecting a significantly slower and more constrained system response. This structural change in the closed-loop dynamics is accompanied by a severe degradation of robustness margins. The corresponding gain margin is $7.56\;\mathrm{dB},$ and the phase margin is $44.5^\circ$, indicating that the system operates very close to the boundary of stability. This behavior is again fully consistent with the observations made in the simplified case: once the safety filter becomes active, the effective feedback law is altered, resulting in a state-dependent modification of the closed-loop dynamics and a corresponding reduction in robustness. The quantitative comparison is summarized in Table~\ref{tab:poles_margins_missile}.




\begin{table*}[t]
\centering
\caption{Comparison of local closed-loop pole locations and disk-based robustness margins for nominal, inactive, and active safety-filter configurations.}
\label{tab:poles_margins_missile}
\begin{tabular}{lccc}
\toprule
\textbf{Case} & \textbf{Gain Margin (dB)} & \textbf{Phase Margin (deg)} \\
\midrule
Nominal 
& $\pm 16.400$ 
& $\pm72.7877$ \\

Inactive CLM-CBF 
& $\pm 16.400$ 
& $\pm72.7877$ \\

Active CLM-CBF 
& $\pm 7.56$ 
& $\pm 44.5$ \\

\bottomrule
\end{tabular}
\end{table*}

A simulation study examines the dynamic behavior under a continuously varying unsafe reference. The corresponding results are presented in Fig.~\ref{fig:time_series_comparison}. In this case, a sinusoidal reference command with an amplitude exceeding the admissible envelope is applied, thereby repeatedly forcing activation and deactivation of the safety filters. The imposed constraints correspond to the limits introduced in Sect.~\ref{sec:Problem Formulatio}, namely $\alpha \in [-15^\circ,\,15^\circ]$ and $q \in [-30^\circ/\text{s},\,30^\circ/\text{s}]$, together with actuator magnitude and rate limits of $\delta \in [-30^\circ,\,30^\circ]$ and $\dot{\delta} \in [-90^\circ/\text{s},\,90^\circ/\text{s}]$. This scenario therefore provides a comprehensive assessment of constraint enforcement under dynamic excitation.

\begin{figure*}[ht]
    \centering
    \includegraphics[width=\textwidth]{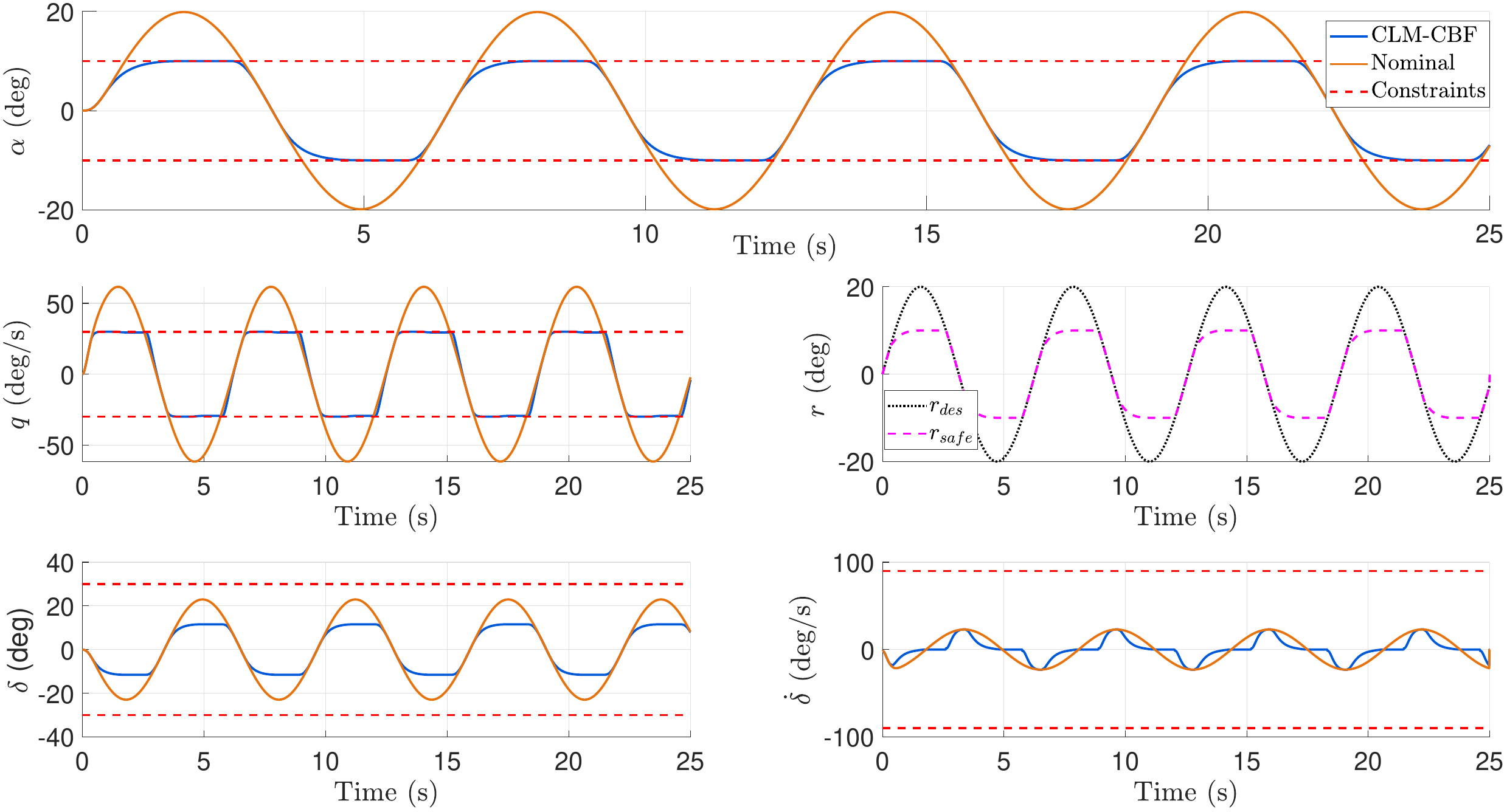}
    \caption{Time-series comparison of nominal control and CLM-CBF under a sinusoidal reference command. The plots show the angle of attack $\alpha(t)$, pitch rate $q(t)$, and fin deflection $\delta(t)$, together with their admissible bounds.}
    \label{fig:time_series_comparison}
\end{figure*}

The time-series results in Fig.~\ref{fig:time_series_comparison} compare the nominal controller and the CLM-CBF formulation. It can be observed that the nominal system response in orange violates the angle-of-attack constraints whenever the reference exceeds the admissible range. In contrast, the response of the closed-loop system with a CLM-CBF-based safety filter consistently enforces the prescribed bounds over the entire simulation horizon. The safety filter effectively reshapes the unsafe sinusoidal reference into a feasible trajectory, ensuring that the resulting angle-of-attack response remains strictly within the admissible set at all times. A similar behavior is observed for the pitch-rate dynamics. While the nominal system exceeds the prescribed bounds on $q(t)$, the CLM-CBF guarantees strict satisfaction of the pitch-rate constraints. Importantly, the safety filter operates on the intersection of both constraints, simultaneously enforcing limits on $\alpha(t)$ and $q(t)$ without violation. Furthermore, the actuator signals remain within the specified magnitude and rate limits, demonstrating that the choice of $\gamma_{\alpha,u}$, $\gamma_{\alpha,l}$, $\gamma_{q,u}$ and $\gamma_{q,l}$ within the QP formulation is appropriate.

Overall, these results validate the theoretical developments of Section~\ref{sec:Closed_Loop_Model_CBFs}. For the considered linear missile dynamics and affine state-feedback structure, the proposed CLM-CBF is not merely an alternative implementation of safety filtering, but a formulation that is behaviorally equivalent to the classical OLM-CBF. This is a practically relevant result, since it shows that safety can be enforced at the reference-level without any loss in performance, safety, or closed-loop fidelity. Consequently, the choice between both formulations can be made based on architectural considerations, such as modularity, integration into hierarchical controllers, or accessibility of actuator commands, rather than on differences in safety performance.

%% file: Sections/Conclusion.tex
This paper introduced a CLM-CBF framework for flight envelope protection. In contrast to conventional OLM-CBF approaches, which enforce safety at the control input-level and may alter the closed-loop dynamics, the proposed method enforces safety at the reference-level using an explicit representation of the closed-loop system. This perspective enables constraint enforcement while preserving critical phase and gain margins of the designed control system. In particular, the framework allows safety filtering to be integrated without modifying the controller output, thereby maintaining the desired closed-loop behavior and enabling modular integration into existing flight control architectures.

The effectiveness of the approach was demonstrated on a nonlinear longitudinal missile model with realistic aerodynamic and actuator characteristics. Simulation results show that the method reliably prevents constraint violations during aggressive maneuvers while maintaining nominal tracking performance when no safety constraints are active. The results further demonstrate that the CLM-CBF-based safety filter achieves constraint enforcement performance comparable to conventional OLM-CBF approaches, while preserving the nominal closed-loop behavior of the underlying control system. 

Future work will focus on extending the framework to more complex scenarios, including multi-dimensional flight dynamics, tighter uncertainty bounds, and experimental validation on real-world systems.

%% file: Database_Bibliography_Upd.bib
@inproceedings{Grondman2018,
  author = {Grondman, F. and Looye, G. and Kuchar, R. O. and Chu, Q. P. and van Kampen, E.-J.},
  title = {Design and Flight Testing of Incremental Nonlinear Dynamic Inversion-based Control Laws for a Passenger Aircraft},
  booktitle = {AIAA Guidance, Navigation, and Control Conference},
  year = {2018},
  note = {AIAA-2018-0385}
}

@INPROCEEDINGS{Stougie2024-ym,
  title     = "Incremental nonlinear dynamic inversion control with flight
               envelope protection for the flying-{V}",
  author    = "Stougie, Jurian and Pollack, Tijmen and Van Kampen, Erik-Jan",
  booktitle = "AIAA SCITECH 2024 Forum",
  publisher = "American Institute of Aeronautics and Astronautics",
  address   = "Reston, Virginia",
  month     =  jan,
  year      =  2024
}

@article{ames2016control, 

  title={Control barrier function based quadratic programs for safety critical systems}, 

  author={Ames, Aaron D and Xu, Xiangru and Grizzle, Jessy W and Tabuada, Paulo}, 

  journal={IEEE Transactions on Automatic Control}, 

  volume={62}, 

  number={8}, 

  pages={3861--3876}, 

  year={2016}, 

  publisher={IEEE} 

}

@INPROCEEDINGS{Ames_2014,
author={Ames, Aaron D. and Grizzle, Jessy W. and Tabuada, Paulo},
booktitle={53rd IEEE Conference on Decision and Control}, 
title={Control barrier function based quadratic programs with application to adaptive cruise control}, 
year={2014},
volume={},
number={},
pages={6271-6278},
doi={10.1109/CDC.2014.7040372      }}


%% file: sample.bib
@article{Autenrieb_2024b,
author = {Autenrieb, Johannes and Shin, Hyo-Sang},
title = {Complementary Filter-Based Incremental Nonlinear Model Following Control Design for a Tilt-Wing UAV},
journal = {International Journal of Robust and Nonlinear Control},
volume = {35},
number = {4},
pages = {1596-1615},
doi = {https://doi.org/10.1002/rnc.7743},
year = {2025}
}

@INPROCEEDINGS{Oudin2017,
author = {Simon Oudin},
title = {Low Speed Protections for a Commercial Airliner: a Practical Approach},
booktitle = {AIAA Guidance, Navigation, and Control Conference},
chapter = {},
pages = {},
year      = {2017},
doi = {10.2514/6.2017-1023}
}

@misc{Fisher2026,
      title={An Error-Based Safety Buffer for Safe Adaptive Control (Extended Version)}, 
      author={Peter A. Fisher and Johannes Autenrieb and Anuradha M. Annaswamy},
      year={2026},
      eprint={2510.23491},
      archivePrefix={arXiv},
      primaryClass={eess.SY},
      url={https://arxiv.org/abs/2510.23491}, 
}

@article{Garone2017,
title = {Reference and command governors for systems with constraints: A survey on theory and applications},
journal = {Automatica},
volume = {75},
pages = {306-328},
year = {2017},
issn = {0005-1098},
doi = {https://doi.org/10.1016/j.automatica.2016.08.013},
author = {Emanuele Garone and Stefano {Di Cairano} and Ilya Kolmanovsky}
}

@INPROCEEDINGS{Kolmanovsky2014,
  author={Kolmanovsky, Ilya and Garone, Emanuele and Di Cairano, Stefano},
  booktitle={2014 American Control Conference}, 
  title={Reference and command governors: A tutorial on their theory and automotive applications}, 
  year={2014},
  volume={},
  number={},
  pages={226-241},
  doi={10.1109/ACC.2014.6859176}}

@ARTICLE{Autenrieb2024,
  title     = "Flight control design for a hypersonic waverider configuration: A
               non-linear model following control approach",
  author    = "Autenrieb, Johannes and Fezans, Nicolas",
  journal   = "CEAS Space Journal",
  publisher = "Springer Science",
  pages     = "1--24",
  month     =  apr,
  year      =  2024,
  language  = "en"
}

@INPROCEEDINGS{Autenrieb2023b,
  author={Autenrieb, Johannes and Annaswamy, Anuradha},
  booktitle={2023 62nd IEEE Conference on Decision and Control (CDC)}, 
  title={Safe and Stable Adaptive Control for a Class of Dynamic Systems}, 
  year={2023},
  volume={},
  number={},
  pages={5059-5066},
  doi={10.1109/CDC49753.2023.10383779}
}

@article{XU2015,
title = {Robustness of Control Barrier Functions for Safety Critical Control},
journal = {IFAC-PapersOnLine},
volume = {48},
number = {27},
pages = {54-61},
year = {2015},
note = {{A}nalysis and {D}esign of {H}ybrid {S}ystems (ADHS)},
issn = {2405-8963},
doi = {https://doi.org/10.1016/j.ifacol.2015.11.152},
author = {Xiangru Xu and Paulo Tabuada and Jessy W. Grizzle and Aaron D. Ames}
}

@ARTICLE{Hwang_2017,
author = {{Hwang}, Donghyeok and {Tahk}, Min-Jea},
title = "{The Inverse Optimal Control Problem for a Three-Loop Missile Autopilot}",
journal = {International Journal of Aeronautical and Space Sciences},
year = 2018,
month = jun,
volume = {19},
number = {2},
doi = {10.1007/s42405-018-0014-6     },
note = {Provided by the SAO/NASA Astrophysics Data System}
}

@article{Nagumo_1942,
title={{\"U}ber die Lage der Integralkurven gew{\"o}hnlicher Differentialgleichungen},
author={Mitio Nagumo},
journal={Proceedings of the Physico-Mathematical Society of Japan. 3rd Series},
volume={24},
number={ },
pages={551-559},
year={1942},
doi={10.11429/ppmsj1919.24.0_551}
}

@inproceedings{Seo2017,
  author    = {Yongjun Seo and Youdan Kim},
  title     = {Robust Control Augmentation System for Flight Envelope Protection Using Backstepping Control Scheme},
  booktitle = {7th European Conference for Aeronautics and Aerospace Sciences (EUCASS)},
  year      = {2017},
  address   = {Milan, Italy},
  publisher = {EUCASS},
  doi       = {10.13009/EUCASS2017-192}
}

@inproceedings{Steffensen2019,
    author = {Steffensen, Rasmus and Gabrys, Agnes and Holzapfel, Florian},
    title = {Flight Envelope Protections Using Phase Plane Limits and Backstepping Control},
    booktitle = {CEAS EuroGNC Conference 2019},
    address = {Milan, Italy},
    month = apr,
    year = {2019},
    note = {CEAS-GNC-2019-003}
}

@article{Lombaerts2017,
author = {Lombaerts, Thomas and Looye, Gertjan and Ellerbroek, Joost and Martin, Mitchell},
year = {2017},
month = {08},
number = {8},
pages = {1902-1924},
title = {Design and Piloted Simulator Evaluation of Adaptive Safe Flight Envelope Protection Algorithm},
volume = {40},
journal = {Journal of Guidance Control and Dynamics},
doi = {10.2514/1.G002525}
}

@inproceedings{Tang2009,
author = {Tang, Liang and Roemer, Michael and Ge, Jianhua and Crassidis, Agamemnon and Prasad, Jvr and Belcastro, Christine},
year = {2009},
month = {08},
pages = {},
title = {Methodologies for Adaptive Flight Envelope Estimation and Protection},
isbn = {978-1-60086-978-5},
booktitle = {AIAA Guidance, Navigation, and Control Conference and Exhibit},
doi = {10.2514/6.2009-6260}
}

@article{Falkena2010,
author = {Falkena, Wouter and Borst, Clark and Mulder, J.A.},
year = {2010},
month = {08},
pages = {},
title = {Investigation of Practical Flight Envelope Protection Systems for Small Aircraft},
volume = {34},
isbn = {978-1-60086-962-4},
journal = {Journal of Guidance, Control, and Dynamics},
doi = {10.2514/6.2010-7701}
}

@INPROCEEDINGS{Autenrieb2025ACC,
  author={Autenrieb, Johannes and Shin, Hyo-Sang},
  booktitle={2025 American Control Conference (ACC)}, 
  title={Sensor-Based Safety-Critical Control Using an Incremental Control Barrier Function Formulation via Reduced-Order Approximate Models}, 
  year={2025},
  volume={},
  number={},
  pages={374-381},
  doi={10.23919/ACC63710.2025.11107913}}

@article{Autenrieb_2025,
author = {Autenrieb, Johannes},
title = {Quadratic Programming Approach to Flight Envelope Protection Using Control Barrier Functions},
journal = {Journal of Guidance, Control, and Dynamics},
volume = {0},
number = {0},
pages = {1-12},
year = {0},
doi = {10.2514/1.G009203}

}

@inproceedings{apkarian2014multi,
  title={Multi-model, multi-objective tuning of fixed-structure controllers},
  author={Apkarian, Pierre and Gahinet, Pascal and Buhr, Craig},
  booktitle={2014 European Control Conference (ECC)},
  pages={856--861},
  year={2014},
  organization={IEEE}
}
